\documentclass[12pt,authoryear,review]{elsarticle}
\usepackage[latin9]{inputenc}
\usepackage[left=1.25in,right=1.25in,top=1in,bottom=1in,a4paper]{geometry}
\usepackage{url}
\usepackage{amsmath}
\usepackage{mathrsfs}
\usepackage{amsthm}
\usepackage{amssymb}
\usepackage{scalefnt}
\usepackage{amsfonts}
\usepackage[singlespacing]{setspace}
\usepackage[bottom]{footmisc}
\usepackage{indentfirst}
\usepackage{endnotes}
\usepackage{afterpage}
\usepackage{caption}
\usepackage{thmtools}
\usepackage{float}
\usepackage{mathpazo}
\usepackage{chngcntr}
\usepackage{apptools}
\usepackage{makecell}

\AtAppendix{\counterwithin{definition}{section}}
\AtAppendix{\counterwithin{lemma}{section}}
\AtAppendix{\counterwithin{proposition}{section}}

\usepackage[unicode=true,
 bookmarks=false,
 breaklinks=false,pdfborder={0 0 1}, colorlinks=false]
 {hyperref}
\hypersetup{
 colorlinks=true, linkcolor= blue, citecolor=blue, urlcolor=blue, pdfauthor=author}

\theoremstyle{plain}
\newtheorem{theorem}{Theorem}
\newtheorem{corollary}{Corollary}
\newtheorem{assumption}{Assumption}
\newtheorem{definition}{Definition}
\newtheorem{example}{Example}
\newtheorem{lemma}{Lemma}
\newtheorem{proposition}{Proposition}
\newtheorem{remark}{Remark}

\DeclareMathOperator*{\argmax}{arg\,max}

\newcommand\cites[1]{\citeauthor{#1}'s\ (\citeyear{#1})}\newcommand{\com}[1]{}

\begin{document}

\begin{frontmatter}{}

\title{Revenue Comparisons of Auctions\\with Ambiguity Averse Sellers}

\author[A1]{Sosung Baik}

\ead{sosung.baik@gmail.com}

\author[A1]{Sung-Ha Hwang\corref{cor1}}

\ead{sungha@kaist.ac.kr}

\cortext[cor1]{\today. Corresponding author. The research of S.-H. H. was supported by
the National Research Foundation of Korea Grant funded by the Korean
Government.}

\address[A1]{Korea Advanced Institute of Science and Technology (KAIST), Seoul, Korea}

\begin{abstract}

We study the revenue comparison problem of auctions when the seller has a maxmin expected utility preference. The seller holds a set of priors around some reference belief, interpreted as an approximating model of the true probability law or the focal point distribution. We develop a methodology for comparing the revenue performances of auctions: the seller prefers auction $X$ to auction $Y$ if their transfer functions satisfy a weak form of the single-crossing condition. Intuitively, this condition means that a bidder's payment is more negatively associated with the competitor's type in $X$ than in $Y$. Applying this methodology, we show that when the reference belief is independent and identically distributed (IID) and the bidders are ambiguity neutral, (i) the first-price auction outperforms the second-price and all-pay auctions, and (ii) the second-price and all-pay auctions outperform the war of attrition. Our methodology yields results opposite to those of the Linkage Principle.

\end{abstract}
\end{frontmatter}

\thispagestyle{empty}

\newpage
\section{Introduction\setcounter{page}{1}}

Since the establishment of the Revenue Equivalence Principle \citep{Myer81}, an important problem of auction theory is to compare the revenue performances of different auctions in setups relaxing \cites{Myer81} standard assumptions. The Linkage Principle \citep{Mil82, Kris97}, one of the fundamental results in this direction, states that in the affiliated interdependent values setup, auctions with stronger positive linkages between a bidder's payment and her own signal yield higher expected revenues. Succeeding works study the effects of the bidders' risk aversion \citep{Mas84}, the seller's risk aversion \citep{Wae98}, the bidders' financial constraints \citep{Che98}, and asymmetric valuation distributions \citep{Mas00}. 

This paper studies the revenue comparison problem in which the seller's preference exhibits ambiguity aversion \citep{Ell61}. One of our main contributions, Theorem \ref{thm:methodology}, provides a methodology to compare the revenue performances of different auctions. Intuitively, it states that auctions in which a bidder's payment is more negatively associated with the competitor's type yield higher revenues. Applying this methodology, we compare the revenues of four commonly studied auctions: the first-price, second-price, all-pay auctions and war of attrition.

Following the maxmin expected utility model \citep[MMEU;][]{Gil89}, the seller holds a set of priors around some \textit{reference belief} and evaluates an auction by the \textit{worst-case revenue}, the minimum expected revenue over the set of priors. The reference belief can be interpreted as an approximation of the true distribution \citep{Han01, Han08} or the focal point distribution \citep{Bo06, Bo09}. To present our results clearly, we focus primarily on the case of ambiguity neutral bidders; however, most of our results extend to the case of ambiguity averse bidders (Section \ref{subsec:bidder_aversion}).

To develop our main methodology, we first show that in finding the beliefs that minimize the seller's expected revenue, we can restrict attention to a special class of beliefs within the set of priors named the \textit{decreasing rearrangements} (Theorem \ref{thm:restriction}). A rearrangement of a belief reassigns probabilities over the state space; it is called a decreasing rearrangement if it overweights the likelihood of low types and underweights that of high types relative to the reference belief (Definition \ref{def:rearrangement}). Since facing low types is an unfavorable event and facing high types is a favorable event for the seller, Theorem \ref{thm:restriction} is intuitive. To confirm this intuition, for a given belief, we explicitly construct its decreasing rearrangement that yields a lower expected revenue than the original belief. To ensure that this decreasing rearrangement lies in the set of priors, we assume that the set of priors is \textit{rearrangement invariant}, i.e., it is closed under the rearrangement operation (Assumptions \ref{assum:Q}-\ref{assum:Q_indep_IID}). This assumption is satisfied by 
a wide range of sets of priors used in the literature---most importantly, the relative entropy neighborhood \citep[Example \ref{ex:Q};][]{Han01, Han08}.

Building on Theorem \ref{thm:restriction}, Theorem \ref{thm:methodology} states that the seller prefers auction $X$ to auction $Y$ if the following two conditions hold. First, each type of bidder's payment is greater (or smaller) in $X$ than in $Y$ against a competitor with low types (or high types) (\textit{Weak Single-Crossing Condition, WSCC}; Figure \ref{fig:single_crossing_condition}). This condition is a weak form of the standard \textit{single-crossing condition} (SCC) in auction theory \citep{Mil04}; hence the name WSCC. Intuitively, WSCC means that a bidder's payment is more negatively associated with the competitor's type in $X$ than in $Y$. Second, $X$ yields at least as high interim expected revenues as $Y$ under the reference belief (\textit{Reference Revenue Condition, RRC}). In applications, the second condition automatically holds as  equality by the Revenue Equivalence Principle \citep{Myer81}. Thus, Theorem \ref{thm:methodology} shows that auctions with stronger negative associations between a bidder's payment and her competitor's type yield higher worst-case revenues.

The intuition behind Theorem \ref{thm:methodology} is as follows. By WSCC, a bidder's payment is greater (or smaller) in $X$ than in $Y$ against a competitor with low types (or high types). However, a decreasing rearrangement overweights the likelihood of low types and underweights that of high types relative to the reference belief. Hence, it overweights the event that the bidder's payment is greater in $X$ than in $Y$, and underweights the opposite event. This, together with RRC, implies that under any decreasing rearrangement, $X$ yields a higher expected revenue than $Y$. By Theorem \ref{thm:restriction}, the worst-case revenue---the minimum over the decreasing rearrangements within the set of priors---is higher in $X$ than in $Y$.

Then, applying Theorem \ref{thm:methodology}, we establish the worst-case revenue rankings between the four commonly studied auctions (Theorem \ref{thm:ranking} and Figures \ref{fig:summary}-\ref{fig:payment_schedule_comparison}). We find that when the reference belief is independent and identically distributed (IID) and the bidders are ambiguity neutral, (i) the first-price auction outperforms the second-price and all-pay auctions, and (ii) the second-price and all-pay auctions outperform the war of attrition. The ranking between the second-price and all-pay auctions is indeterminate (Figure \ref{fig:SPA_APA_indeterminate}).

Notably, the worst-case revenue rankings in Theorem \ref{thm:ranking} are opposite to the expected revenue rankings in the affiliated values setup \citep{Mil82, Kris97}. This is because Theorem \ref{thm:methodology} works in the opposite direction to the Linkage Principle (Proposition \ref{prop:linkage_4}). Recall that according to Theorem \ref{thm:methodology}, if a bidder's payment is more negatively associated with the competitor's type in auction $X$ than in auction $Y$ (WSCC), then $X$ outperforms $Y$. By contrast, the Linkage Principle states that if a bidder's payment is more negatively associated with her own type in $X$ than in $Y$ (Linkage Condition, LC; Theorem \ref{thm:linkage}), then $Y$ outperforms $X$. However, a negative association between a bidder's payment and her competitor's type creates a negative association between her payment and her own type in the affiliated values setup. As a result, WSCC and LC hold simultaneously, and thus the two principles predict opposite results. This logic also implies that in the presence of both ambiguity and affiliation, the rankings between the four auctions are indeterminate (Figure \ref{fig:affiliation}).

Our paper is related to \cite{Che98} in that a version of the single-crossing condition determines the revenue ranking between auctions. Specifically, \cite{Che98} study the setup in which each bidder has private information about her valuation and budget. They show that if the iso-bid curves of two auctions in the two-dimensional space of valuation and budget satisfy a single-crossing condition, their revenues can be compared. \cite{Wae98} also analyze the setup where the seller is risk averse, a natural benchmark for our study. They show that the first-price auction outperforms the second-price auction because the winner's payment is less variable (in the sense of second-order stochastic dominance) in the first-price auction than in the second-price auction. However, their result relies on the assumption that the loser pays nothing, which is violated in the all-pay auction and war of attrition.

The remainder of this paper is organized as follows. Section \ref{sec:model} presents our setup. Section \ref{sec:methodology} develops our main methodology. As an application, Section \ref{sec:ranking_4} compares the four commonly studied auctions. Section \ref{sec:linkage} discusses the relationship between our methodology and the Linkage Principle. Section \ref{sec:extension} provides two extensions: ambiguity averse bidders and ambiguity seeking seller. Section \ref{sec:discussion} discusses the related literature and concludes the paper.

\section{Model} \label{sec:model}

\subsection{Agents and preferences} \label{subsec:MMEU}

A seller wants to sell an indivisible object to two bidders, denoted by bidders $1$ and $2$. Each bidder has a privately known type $\theta \in \Theta = [0, 1]$ representing her valuation for the object. There is a commonly known \textit{reference belief} $P$, a probability measure on $\Theta^2$. As mentioned in the introduction, $P$ can be interpreted as an approximation of the true probability law \citep{Han01, Han08} or the focal point distribution \citep{Bo06, Bo09}. Assume $P$ has a positive probability density.

The seller, being ambiguity averse, takes into account the possibility that the true probability law differs from the reference belief. Specifically, she holds a set of priors about the joint type distribution $\mathcal{Q}$ around the reference belief, where $P \in \mathcal{Q}$. For technical reasons, assume $\mathcal{Q}$ is weakly compact. As mentioned in the introduction, we focus primarily on the case of ambiguity neutral bidders, in which the bidders believe that types are drawn according to the reference belief $P$. The case of ambiguity averse bidders is discussed in Section \ref{subsec:bidder_aversion}.

For a given auction, let $t_i(\theta, \theta')$ denote bidder $i$'s payment when her type is $\theta$ and her competitor's type is $\theta'$. We call $t = (t_1, t_2): \Theta^2 \rightarrow \mathbb{R}_+^2$ the \textit{transfer function}. Following the MMEU model \citep{Gil89}, the seller evaluates an auction by the \textit{worst-case revenue} $\mathcal{R}(t)$, the minimum expected revenue over the set of priors:
\[
    \mathcal{R}(t) := \min_{Q \in \mathcal{Q}} \iint_{\Theta^2} [t_1(\theta, \theta')+t_2(\theta', \theta)] Q(d\theta, d\theta').
\]

\subsection{Rearrangement}

In this section, we first explain the concepts related to \textit{rearrangement}, and then describe our assumption on the seller's set of priors named \textit{rearrangement invariance}. To this end, consider a probability space $(\Omega, \Sigma, \mu)$. Let $\Delta(\Omega, \mu)$ be the set of all probability measures over $\Omega$ that are absolutely continuous with respect to $\mu$. We introduce the following definition:


\begin{definition} \label{def:rearrangement} 

Let $\mathcal{S} \subset \Delta(\Omega, \mu)$.

\noindent (i) We say $\nu' \in \Delta(\Omega, \mu)$ is a \textbf{rearrangement} of $\nu \in \Delta(\Omega, \mu)$ (with respect to $\mu$) if
\begin{equation} \label{eq:def_rearrangement}
\mu \{ \omega: \frac{d\nu'}{d\mu}(\omega) \leq c\} = \mu \{ \omega: \frac{d\nu}{d\mu}(\omega) \leq c\} \quad \text{for all $c \geq 0$.}
\end{equation}

\noindent (ii) Suppose $\Omega$ is endowed with a partial order $\leq$. We say $\nu' \in \Delta(\Omega, \mu)$ is a \textbf{decreasing rearrangement} of $\nu \in \Delta(\Omega, \mu)$ (with respect to $\mu$) if $\nu'$ is a rearrangement of $\nu$, and
\begin{equation} \label{eq:def_dec_rearrangement}
\omega \leq \omega' \implies \frac{d\nu'}{d\mu}(\omega) \geq \frac{d\nu'}{d\mu}(\omega').\footnote{When $\Omega$ is an interval in $\mathbb{R}$, condition \eqref{eq:def_dec_rearrangement} holds if and only if $\nu'$ is likelihood ratio dominated by $\mu$. However, in higher dimensions, condition \eqref{eq:def_dec_rearrangement} is weaker than likelihood ratio dominance \citep[see][Sec. 6.E]{Sha94}.}
\end{equation}
Whenever $\Omega$ is a product space, $\leq$ is assumed to be the componentwise order. In this case, condition \eqref{eq:def_dec_rearrangement} means that $d\nu/d\mu$ decreases in each argument.\footnote{Throughout this paper, ``increasing'' means ``non-decreasing'' and ``decreasing'' means ``non-increasing''.}

\noindent (iii) For $\mathcal{Q} \subset \mathcal{S}$, we say $\mathcal{Q}$ is \textbf{rearrangement invariant} relative to $\mathcal{S}$ (with respect to $\mu$) if whenever a belief belongs to $\mathcal{Q}$, its rearrangements in $\mathcal{S}$ belong to $\mathcal{Q}$: i.e.,
\[
\text{$\nu \in \mathcal{Q}$, \, and \, $\nu' \in \mathcal{S}$ is a rearrangement of $\nu$} \implies \nu' \in \mathcal{Q}.
\]
In particular, if $\mathcal{S} = \Delta(\Omega, \mu)$, we simply say $\mathcal{Q}$ is rearrangement invariant.

\end{definition}

\begin{figure} [t]
\centering
\includegraphics[scale=.95]{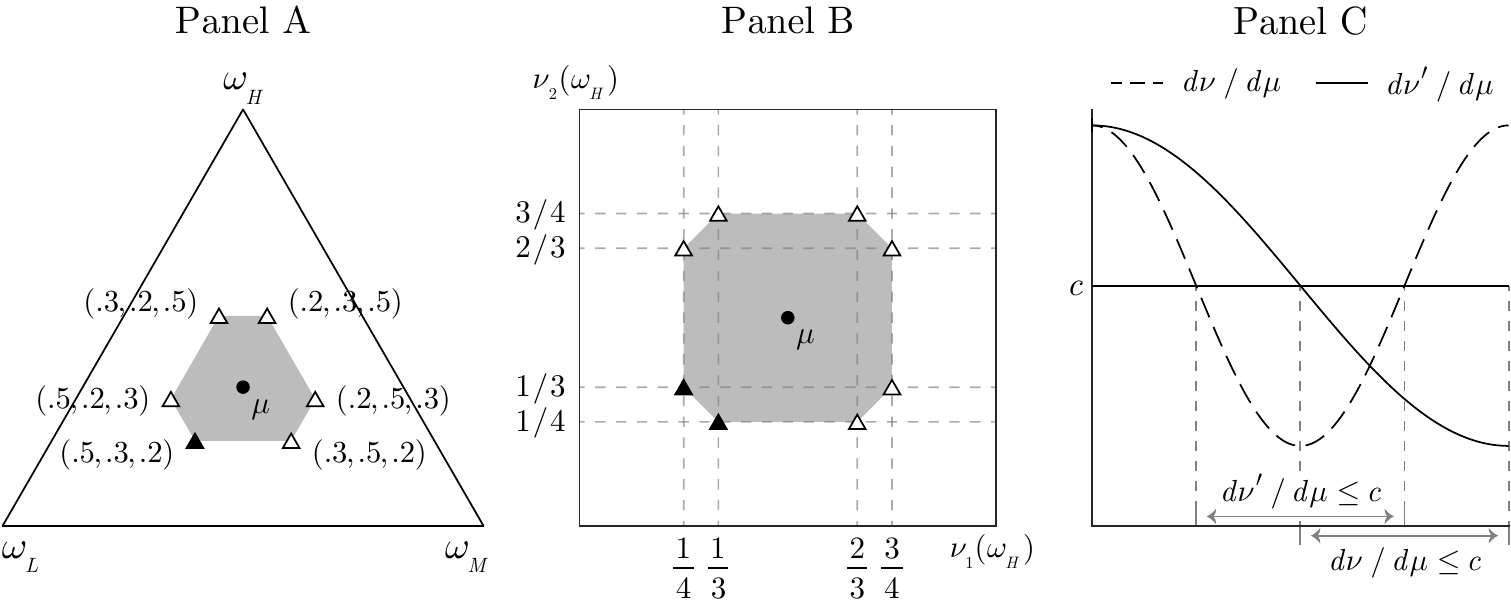}
\caption{\textbf{Rearrangements.} \\
\textbf{Panel A: Discrete states, the domain of all beliefs.} Let $\Omega = \{\omega_L, \omega_M, \omega_H\}$ (where $\omega_L < \omega_M < \omega_H$), $\mu$ be uniform on $\Omega$, and $\mathcal{S} = \Delta(\Omega, \mu)$. A belief $\nu \in \mathcal{S}$ is represented by point $(\nu(\omega_L), \nu(\omega_M), \nu(\omega_H))$ on the simplex. Each $\nu \in \mathcal{S}$ has $3!=6$ rearrangements, marked by triangles. Among them, the black triangle is the decreasing rearrangement. The shaded hexagon is rearrangement invariant. \\
\textbf{Panel B: Discrete states, the domain of independent beliefs.} Let $\Omega = \{\omega_L, \omega_H\} \times \{\omega_L, \omega_H\}$ (where $\omega_L < \omega_H$), $\mu$ be uniform on $\Omega$, and $\mathcal{S}$ be the set of independent beliefs on $\Omega$. A belief $\nu \in \mathcal{S}$ is represented by point $(\nu_1(\omega_H), \nu_2(\omega_H)) \in [0,1]^2$, where $\nu_i$ denotes the $i$-th marginal probability measure of $\nu$. Each $\nu \in \mathcal{S}$ has $4!=24$ rearrangements; out of these, $8$ beliefs marked by triangles are independent, i.e., lie in $\mathcal{S}$. Among them, the two black triangles are the decreasing rearrangements. The shaded octagon is rearrangement invariant relative to $\mathcal{S}$. \\ 
\textbf{Panel C: Continuous states.} Let $\Omega = [0, 1]$ and $\mu$ be uniform over $\Omega$. Suppose that the Radon-Nikodym derivatives of $\nu, \nu' \in \Delta(\Omega, \mu)$ are given as in the figure. Since the lower contour sets of $d\nu/d\mu$ and $d\nu'/d\mu$ have the same length, condition \eqref{eq:def_rearrangement} holds. Because $d\nu'/d\mu$ is decreasing, $\nu'$ is a decreasing rearrangement of $\nu$.} \label{fig:rearrangement_def}

\end{figure}

Figure \ref{fig:rearrangement_def} illustrates Definition \ref{def:rearrangement}. To explain Definition \ref{def:rearrangement}, consider the simple case of discrete $\Omega$ and uniform $\mu$. Then, the rearrangements are equivalent to the permutations of probabilities over states (Panels A-B, triangles). Among them, the decreasing rearrangements are the permutations that assign high probabilities to low states and low probabilities to high states (Panels A-B, black triangles). Accordingly, rearrangement invariance requires that a set of priors remains unchanged under permutations (Panels A-B, shaded regions). This means that the set of priors---and hence the degree of ambiguity it represents---is independent of the specific ordering of states. Definition \ref{def:rearrangement} extends these concepts to general state spaces (including continuous spaces; Panel C).

One of our main results in Section \ref{sec:methodology} states that when the set of priors is rearrangement invariant, in finding the beliefs that minimize the seller's expected revenue, it suffices to focus on the decreasing rearrangements within the set of priors (Theorem \ref{thm:restriction}). To prove this result, for a given belief in the set of priors, we construct its decreasing rearrangement which yields a lower expected revenue than the original belief. The rearrangement invariance property ensures that this decreasing rearrangement lies in the set of priors.

The rearrangement invariance in Definition \ref{def:rearrangement} is also closely related to a property known as \textit{probabilistic sophistication} \citep{Machi92, Ghi02, Macc06, Cer11, Cer12}.\footnote{This property is also called the \textit{neutrality axiom} in the literature on probabilistic risk aversion \citep{Yaa87, Saf98}.} An MMEU decision maker with a set of priors $\mathcal{Q}$ is said to be probabilistically sophisticated if the following holds: for all bounded measurable functions $T, T': \Omega \rightarrow \mathbb{R}$,
\begin{align}
&\mu \{ \omega: T(\omega) \leq c \} = \mu \{ \omega: T'(\omega) \leq c \} \quad \text{for all} \, c \in \mathbb{R} \label{eq:prob_soph_equal_dtn} \\
&\implies \min_{\nu \in \mathcal{Q}} \iint_{\Omega} T d\nu = \min_{\nu \in \mathcal{Q}} \iint_{\Omega} T' d\nu. \label{eq:prob_soph_equal_util}
\end{align}
That is, if two acts $T$ and $T'$ have the same outcome distribution under the reference belief $\mu$, the decision maker is indifferent between them.

\citet[][Thm. 14]{Macc06} show that $\mathcal{Q}$ is rearrangement invariant if and only if the decision maker is probabilistically sophisticated. To illustrate the ``only if'' direction, suppose again that $\Omega$ is discrete and $\mu$ is uniform, and let $\mathcal{Q}$ be rearrangement invariant. It can be shown that if two acts $T$ and $T'$ satisfy condition \eqref{eq:prob_soph_equal_dtn}, then $T'$ is a permutation of $T$. This implies that the minimum expectation of $T'$ over $\mathcal{Q}$ can be expressed as that of $T$ over a permutation of $\mathcal{Q}$. By rearrangement invariance, the permutation of $\mathcal{Q}$ coincides with $\mathcal{Q}$, establishing equation \eqref{eq:prob_soph_equal_util}.\footnote{More precisely, condition \eqref{eq:prob_soph_equal_dtn} implies that there exists a permutation $\sigma$ over $\Omega$ satisfying $T' = T \circ \sigma$. Then, it can be shown that $\min_{\nu \in \mathcal{Q}} \int_\Omega T' d\nu = \min_{\nu \in \mathcal{Q} \circ \sigma^{-1}} \int_\Omega T d\nu$, where $\mathcal{Q} \circ \sigma^{-1} := \{ \nu \circ \sigma^{-1}: \nu \in \mathcal{Q} \}$. Rearrangement invariance implies $\mathcal{Q} \circ \sigma^{-1} = \mathcal{Q}$, establishing equation \eqref{eq:prob_soph_equal_util}.} Hence, the decision maker is probabilistically sophisticated.

\medskip

Now, returning to the auction setup, consider three domains of beliefs $\mathcal{S}$:
\begin{align*}
\Delta(\Theta^2, P) &:= \{ Q: \text{$Q$ is a belief over $\Theta^2$ such that $Q \ll P$} \} \\
\Delta^{Ind}(\Theta^2, P) &:= \{ Q \in \Delta(\Theta^2, P): \text{$Q$ is independent, i.e, $Q = Q_1 \times Q_2$} \}\\
\Delta^{IID}(\Theta^2, P) &:= \{ Q \in \Delta(\Theta^2, P): \text{$Q$ is IID, i.e, $Q= Q_1 \times Q_2$ and  $Q_1=Q_2$} \},
\end{align*}
where $Q_i$ denotes the $i$-th marginal probability measure of $Q$. Observe that $\Delta(\Theta^2, P) \supset \Delta^{Ind}(\Theta^2, P) \supset \Delta^{IID}(\Theta^2, P)$.

Our first assumption on the seller's set of priors $\mathcal{Q}$ is as follows:
\begingroup
\renewcommand\theassumption{1A}
\begin{assumption} \label{assum:Q}
For $\mathcal{S} = \Delta(\Theta^2, P)$, the following holds:

\noindent (i) $\mathcal{Q} \subset \mathcal{S}$.

\noindent (ii) $\mathcal{Q}$ is rearrangement invariant with respect to $\mathcal{S}$.
\end{assumption}
\endgroup

\noindent This assumption holds for a wide range of sets of priors used in the literature, as shown in Example \ref{ex:Q}.

\begin{example} [Set of priors] \label{ex:Q} \normalfont \,

\noindent \textbf{(a) Divergence neighborhood \citep{Han01, Han08}.}

\noindent The $\phi$-divergence is a measure of discrepancy between probability measures used in information theory and statistics \citep{Ali66}. Given a convex function $\phi: \mathbb{R}_+ \rightarrow \mathbb{R}$, the $\phi$-divergence is defined as follows: for probability measures $\mu$ and $\nu$ on the same state space,
\[
D(\nu||\mu) := \int \phi \left(\frac{d\nu}{d\mu}\right) d\mu \quad \text{if $\nu \ll \mu$,} \quad \text{and} \quad D(\nu||\mu):=\infty \quad \text{otherwise.}
\]
In the special case of $\phi(z) \equiv z \log z$, $\phi$-divergence becomes the popular \textit{relative entropy} \citep{Kull51}.

Now, let $\mathcal{Q}$ be the set of beliefs that are close to the reference belief, where the ``closeness'' is measured by divergence:
\[
    \mathcal{Q} := \{ Q \in \Delta(\Theta^2, P): D(Q||P) \leq \eta \}.
\]
Here, the parameter $\eta \geq 0$ represents the degree of ambiguity. \citet[Thm. 14 and Lem. 15]{Macc06} show that $\mathcal{Q}$ satisfies Assumption \ref{assum:Q}. This is one of the most popular ambiguity sets in the robustness literature \citep{Han01, Han08, Ben13}.

\noindent \textbf{(b) Bounded likelihood ratio \citep{Lo98, Bo06}.}

\noindent Let $\mathcal{Q}$ be the set of beliefs whose likelihood ratios lie in a given interval:
\[
\mathcal{Q} := \{ Q \in \Delta(\Theta^2, P): dQ/dP \in [1-\alpha \eta, 1+\beta \eta] \},
\]
where $\eta \geq 0$ represents the degree of ambiguity and $\alpha, \beta \geq 0$. Because the rearrangement operation preserves the range of $dQ/dP$, $\mathcal{Q}$ satisfies Assumption \ref{assum:Q}. In the limiting case of $\beta = \infty$, $\mathcal{Q}$ reduces to the \textit{contamination model}:
\[
\mathcal{Q} := \{ Q = \eta R + (1-\eta) P: R \in \Delta(\Theta^2, P) \},
\]
where $\alpha$ is normalized to $1$. This model is widely used in the literature on mechanism design with ambiguity \citep{Bo06, Bo09}. $\square$

\end{example}

Some studies suppose that the set of priors consists only of independent beliefs or of IID beliefs \citep[e.g.,][]{Lo98, Bo06}. This corresponds to situations in which the seller has additional information that types are independent or IID. In these cases, because a rearrangement of an independent (or IID) belief is not necessarily independent (or IID), Assumption \ref{assum:Q} does not hold. To address this issue, we assume that the set of priors is rearrangement invariant relative to the domain of independent beliefs, or to the domain of IID beliefs.

\begingroup
\renewcommand\theassumption{1B}
\begin{assumption} \label{assum:Q_indep_IID} 

For $\mathcal{S} = \Delta^{Ind}(\Theta^2, P) \text{ or } \Delta^{IID}(\Theta^2, P)$, the following holds:

\noindent (i) $\mathcal{Q} \subset \mathcal{S}$.

\noindent (ii) $\mathcal{Q}$ is rearrrangement invariant relative to $\mathcal{S}$.

\end{assumption}
\endgroup

We provide two examples that satisfy Assumption \ref{assum:Q_indep_IID}.

\begin{example} [Set of priors: Continued] \label{ex:Indep_IID} \normalfont \,

\noindent The natural analogues of Example \ref{ex:Q} (a) and (b) are given as follows:
\begin{align*}
\textbf{(a-Ind)} \quad \mathcal{Q} &:= \{ Q \in \Delta^{Ind}(\Theta^2, P): D(Q_i||P_i) \leq \eta \quad \text{for $i=1, 2$} \} \\
\textbf{(a-IID)} \quad \mathcal{Q} &:= \{ Q \in \Delta^{IID}(\Theta^2, P): D(Q_1||P_1) \leq \eta \} \\
\textbf{(b-Ind)} \quad \mathcal{Q} &:= \{ Q \in \Delta^{Ind}(\Theta^2, P): dQ_i/dP_i \in [1-\alpha \eta, 1+\beta \eta] \quad \text{for $i=1, 2$} \} \\
\textbf{(b-IID)} \quad \mathcal{Q} &:= \{ Q \in \Delta^{IID}(\Theta^2, P): dQ_1/dP_1 \in [1-\alpha \eta, 1+\beta \eta] \}.
\end{align*}
The last model (b-IID) is used by \cite{Lo98}. $\square$

\end{example}

\section{Main results} \label{sec:methodology}

This section develops a methodology to compare worst-case revenues. Assumption \ref{assum:t} is a common property of most standard auctions:

\setcounter{assumption}{1}
\begin{assumption} \label{assum:t}
(i) The total transfer $t_1(\theta, \theta')+t_2(\theta', \theta)$ increases in each argument. 

\noindent (ii) The probability measure induced by the total transfer from $P$ is atomless: i.e.,
\[
P \{ (\theta, \theta'): t_1(\theta, \theta')+t_2(\theta', \theta) = c \} = 0 \quad \text{for all $c \in \mathbb{R}_+$.}
\]
\end{assumption}

Theorem \ref{thm:restriction} states that for an auction satisfying Assumption \ref{assum:t}, the seller's worst-case revenue equals the minimum expected revenue over the decreasing rearrangements within the set of priors. Thus, in finding the beliefs that minimize the seller's expected revenue, we can restrict attention to the decreasing rearrangements.



\begin{theorem} \label{thm:restriction}

Suppose $\mathcal{Q}$ satisfies Assumption \ref{assum:Q} or \ref{assum:Q_indep_IID}. Let
\[
\mathcal{Q}^* := \{ Q^* \in \mathcal{Q}: \text{$Q^*$ is a decreasing rearrangement of $Q \in \mathcal{Q}$} \}.
\]
Then, for an auction whose transfer function $t$ satisfies Assumption \ref{assum:t},
\[
    \mathcal{R}(t) = \min_{Q^* \in \mathcal{Q}^*} \iint_{\Theta^2} [t_1(\theta, \theta')+t_2(\theta', \theta)] Q^*(d\theta, d\theta').
\]

\end{theorem}
\begin{proof}
See Section \ref{subsec:methodology_proof}.
\end{proof}

\begin{figure} [t]
\centering
\includegraphics[scale=1]{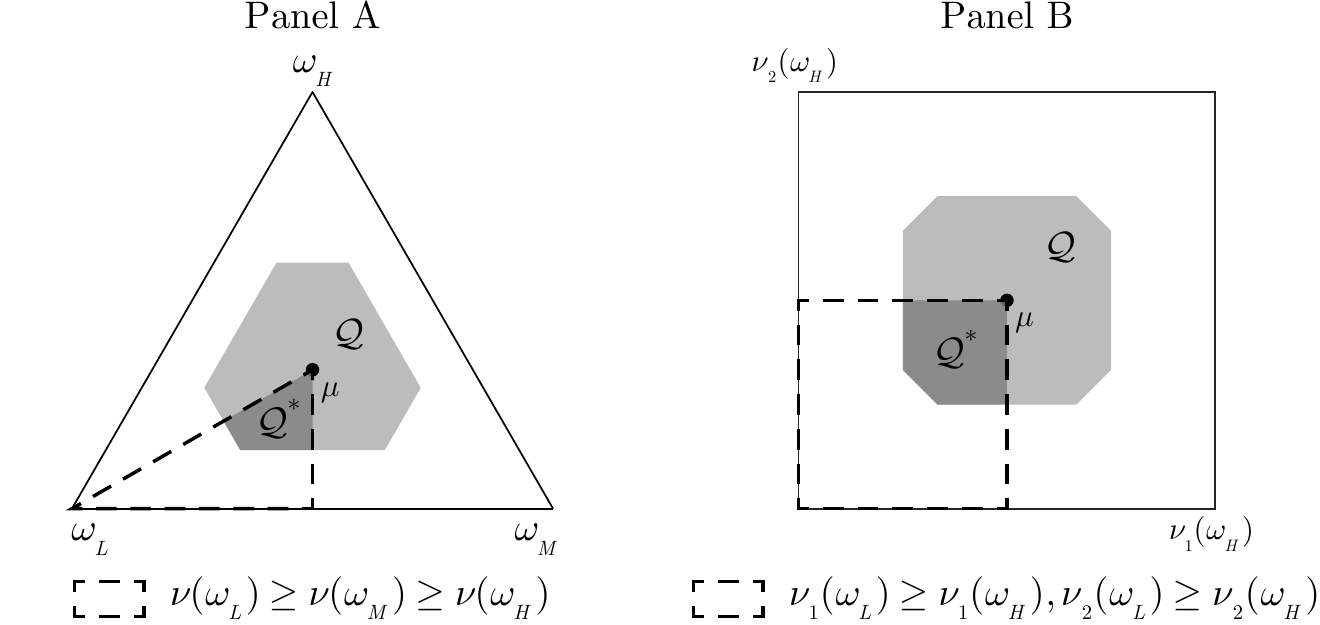}
\caption{\textbf{Decreasing rearrangements.} Let $\Omega$, $\mu$ and $\mathcal{S}$ be given as in Panels A-B of Figure \ref{fig:rearrangement_def}. In each panel, the shaded region represents $\mathcal{Q}$, and the region enclosed by the dashed line represents the set of beliefs $\nu \in \mathcal{S}$ such that $d\nu/d\mu$ is decreasing. The intersection between the two regions corresponds to $\mathcal{Q}^*$ in Theorem \ref{thm:restriction}.} \label{fig:decreasing_rearrangement}
\end{figure}

\noindent Figure \ref{fig:decreasing_rearrangement} illustrates the set of decreasing rearrangements $\mathcal{Q}^*$. Equivalently, $\mathcal{Q}^*$ can be expressed as the set of beliefs with decreasing likelihood ratios $dQ/dP$:
\[
\mathcal{Q}^* = \{ Q^* \in \mathcal{Q}: \text{$\frac{dQ}{dP}$ decreases in each argument} \}.
\]

Building on Theorem \ref{thm:restriction}, Theorem \ref{thm:methodology} states that the seller prefers auction $X$ to auction $Y$ if two conditions hold. First, given $i$ and $\theta$, there exists a threshold $\hat \theta$ such that bidder $i$ of type $\theta$ pays a greater (or smaller) amount in $X$ than in $Y$ against a competitor of type $\theta'<\hat \theta$ (or $\theta'> \hat \theta$) (\textit{Weak Single-Crossing Condition, WSCC}; Figure \ref{fig:single_crossing_condition}). This means that a bidder's payment is more negatively associated with the competitor's type in $X$ than in $Y$. Second, under the reference belief, $X$ yields at least as high interim expected revenues as $Y$ (\textit{Reference Revenue Condition, RRC}). In later applications, this condition automatically holds as an equality by the Revenue Equivalence Principle \citep{Myer81}. Thus, Theorem \ref{thm:methodology} shows that a negative association between a bidder's payment and her competitor's type increases worst-case revenue.

\begin{theorem} \label{thm:methodology} 

Suppose $\mathcal{Q}$ satisfies Assumption \ref{assum:Q} or \ref{assum:Q_indep_IID}. Let $X$ and $Y$ be auctions whose transfers $t^X$ and $t^Y$ satisfy Assumption \ref{assum:t}. Assume the following conditions:

\noindent (i) \textbf{Weak Single-Crossing Condition (WSCC).} For all $i$ and $\theta$, there exists a threshold $\hat \theta \in [0, 1]$ such that
\begin{align*}
\theta' < \hat \theta \implies t^X_i(\theta, \theta') \geq t^Y_i(\theta, \theta'), \quad \text{and} \quad \theta' > \hat \theta \implies t^X_i(\theta, \theta') \leq t^Y_i(\theta, \theta').
\end{align*}

\noindent (ii) \textbf{Reference Revenue Condition (RRC).} For all $i \neq j$ and $\theta$,
\[
\int_\Theta t^X_i(\theta, \theta') P(d\theta'|\theta) \geq \int_\Theta t^Y_i(\theta, \theta') P(d\theta'|\theta),
\]
where $P(\cdot|\theta)$ is the conditional distribution of bidder $j$'s type given bidder $i$'s type $\theta$.\footnote{For notational simplicity, we omit the dependence of the conditional distribution on $(i, j)$.}

\noindent Then,
\[
\mathcal{R}(t^X) \geq \mathcal{R}(t^Y).
\]

\end{theorem}
\begin{proof}
See Appendix \ref{appen:methodology}.
\end{proof}

\begin{figure} [t]
\centering
\includegraphics[scale=1]{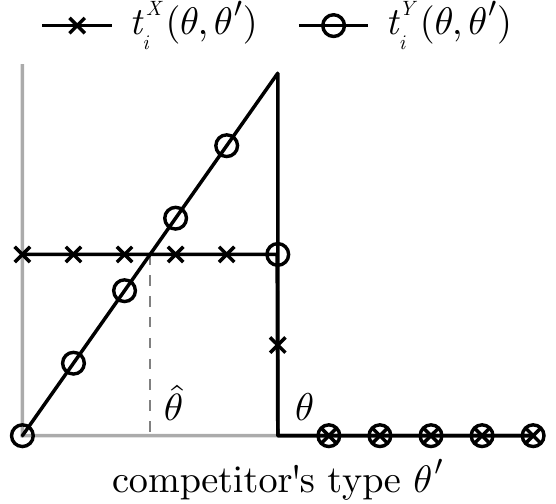}
\caption{\textbf{WSCC (X: first-price auction / Y: second-price auction).} The ``x''-ed and circled lines represent the transfer functions of bidder $i$ with type $\theta$ in $X$ and $Y$, respectively. The horizontal axis represents the competitor's type $\theta'$. Because $t^X_i(\theta, \theta')$ lies weakly above $t^Y_i(\theta, \theta')$ for $\theta' < \hat \theta$ and the opposite holds for $\theta' > \hat \theta$, the pair $(X, Y)$ satisfies WSCC.} \label{fig:single_crossing_condition}
\end{figure}

The intuition of Theorem \ref{thm:methodology} is as follows. To prove this theorem, we establish the following inequality: for all $Q^* \in \mathcal{Q}^*$, $i$ and $\theta$,
\begin{equation} \label{eq:XY_interim_revenue_comparison}
\int_\Theta t^X_i(\theta, \theta') Q^*(d\theta'|\theta) \geq \int_\Theta t^Y_i(\theta, \theta') Q^*(d\theta'|\theta),
\end{equation}
which implies
\[
\min_{Q^* \in \mathcal{Q}^*} \iint_{\Theta^2} [t^X_1(\theta, \theta')+t^X_2(\theta', \theta)] dQ^* \geq \min_{Q^* \in \mathcal{Q}^*} \iint_{\Theta^2} [t^Y_1(\theta, \theta')+t^Y_2(\theta', \theta)] dQ^*.
\]
Then, by Theorem \ref{thm:restriction}, $X$ generates a higher worst-case revenue than $Y$. To show inequality \eqref{eq:XY_interim_revenue_comparison}, let $Q^* \in \mathcal{Q}^*$ be given. By WSCC, bidder $i$ of type $\theta$ pays a greater (or smaller) amount in $X$ than in $Y$ against a competitor with low types (or high types). However, a decreasing rearrangement $Q^*$ overweights the likelihood of low types and underweights that of high types relative to $P$. Thus, $Q^*$ overweights the event that the bidder's payment is greater in $X$ than in $Y$, and underweights the opposite event. Because $X$ yields at least as high interim expected revenues as $Y$ under $P$ by RRC, $X$ yields higher interim expected revenues than $Y$ under $Q^*$. This establishes the desired inequality \eqref{eq:XY_interim_revenue_comparison}.

As mentioned in the introduction, WSCC is a weak form of the \textit{single-crossing condition} familiar from auction theory \citep[][Ch. 4]{Mil04}. Recall that $t^X$ and $t^Y$ satisfy the \textit{single-crossing condition (SCC)} if for all $i$, $\theta$ and $\theta'>\theta''$,
\begin{equation} \label{eq:scc}
    \left\{ \begin{array} {l}
    t_i^X(\theta, \theta'') \leq t_i^Y(\theta, \theta'') \implies t_i^X(\theta, \theta') \leq t_i^Y(\theta, \theta') \\
    t_i^X(\theta, \theta'') < t_i^Y(\theta, \theta'') \implies t_i^X(\theta, \theta') < t_i^Y(\theta, \theta').
    \end{array} \right.
\end{equation}
The first line means that if $t_i^X(\theta, \cdot)$ lies weakly below $t_i^Y(\theta, \cdot)$ at some point $\theta''$, then the same holds at every higher point $\theta'$; the second line is interpreted similarly. Now, WSCC turns out to be equivalent to the following condition (Appendix \ref{appen:scp}): for all $i$, $\theta$ and $\theta'>\theta''$,
\begin{equation} \label{eq:scc_weak}
    t_i^X(\theta, \theta'') < t_i^Y(\theta, \theta'') \implies t_i^X(\theta, \theta') \leq t_i^Y(\theta, \theta').
\end{equation}
This means that if $t_i^X(\theta, \cdot)$ lies \textit{strictly} below $t_i^Y(\theta, \cdot)$ at some point $\theta''$, then $t_i^X(\theta, \cdot)$ lies \textit{weakly} below $t_i^Y(\theta, \cdot)$ at every higher point $\theta'$. It is evident that condition \eqref{eq:scc_weak} is implied by condition \eqref{eq:scc}; hence the name WSCC. Figure \ref{fig:single_crossing_condition} illustrates an example that satisfies WSCC but not SCC.\footnote{Panels B and C of Figure \ref{fig:payment_schedule_comparison} illustrate examples satisfying SCC. Also, Panel D of Figure \ref{fig:payment_schedule_comparison} illustrates another example satisfying WSCC but not SCC.} Like SCC, WSCC requires that $t_i^X(\theta, \cdot)$ crosses $t_i^Y(\theta, \cdot)$ at most once and from above (the point $\hat \theta$). However, WSCC is weaker than SCC in that it allows the two transfer functions to touch outside the crossing point (the interval $[\theta, 1]$).

\subsection{Proof of Theorem \ref{thm:restriction}} \label{subsec:methodology_proof}

This section presents the proof of Theorem \ref{thm:restriction}. We first consider the case of Assumption \ref{assum:Q}, and then Assumption \ref{assum:Q_indep_IID}.

\medskip

\noindent \textit{Case A: $\mathcal{Q}$ satisfies Assumption \ref{assum:Q}.} To prove Theorem \ref{thm:restriction}, we use Proposition \ref{prop:anti_comonotone}, a variant of the Hardy-Littlewood rearrangement inequality \citep{Har59}. Although its proof mostly relies on existing literature \citep[e.g.,][]{Lux67, Follmer16}, we include the proof for completeness.

\begin{proposition} \label{prop:anti_comonotone} 

Let $T: \Theta^2 \rightarrow \mathbb{R}_+$ be measurable and $Q \in \Delta(\Theta^2, P)$. Suppose
\begin{equation} \label{eq:T_atomless}
    P \{ (\theta, \theta'): T(\theta, \theta')=c \} = 0 \quad \text{for all $c \in \mathbb{R}_+$.}
\end{equation}

\noindent (i) There exists a rearrangement $Q_T \in \Delta(\Theta^2, P)$ of $Q$ such that for all $\theta, \theta', \varphi, \varphi' \in \Theta$,
\begin{equation} \label{eq:anti_comonotone_strong}
\left\{ \begin{array}{lll}
    T(\theta, \theta') < T(\varphi, \varphi') &\implies &\frac{dQ_T}{dP}(\theta, \theta') \geq \frac{dQ_T}{dP}(\varphi, \varphi')  \\
    T(\theta, \theta') > T(\varphi, \varphi') &\implies &\frac{dQ_T}{dP}(\theta, \theta') \leq \frac{dQ_T}{dP}(\varphi, \varphi') \\ T(\theta, \theta') = T(\varphi, \varphi') &\implies &\frac{dQ_T}{dP}(\theta, \theta') = \frac{dQ_T}{dP}(\varphi, \varphi').
\end{array} \right.
\end{equation}

\noindent (ii) Moreover, the expectation of $T$ is lower under $Q_T$ than under $Q$:
\[
\iint_{\Theta^2} T(\theta, \theta') Q_T(d\theta, d\theta') \leq \iint_{\Theta^2} T(\theta, \theta') Q(d\theta, d\theta').
\]

\end{proposition}
\begin{proof}
See Appendix \ref{appen:anti_comonotone}.
\end{proof}

Panel C of Figure \ref{fig:rearrangement_ineq_intuition} illustrates $Q_T$. Intuitively, $Q_T$ rearranges $Q$ such that $dQ_T/dP$ varies in exactly the opposite direction to $T$. More precisely, every upper contour set of $dQ_T/dP$ coincides with some lower contour set of $T$. This means that $Q_T$ assigns high probabilities to low values of $T$ and low probabilities to high values of $T$. Thus, the expectation of $T$ is lower under $Q_T$ than under $Q$.

Condition \eqref{eq:anti_comonotone_strong} is a slightly stronger version of the condition known as \textit{anti-comonotonicity} \citep[e.g.,][]{Ghossoub15}, which requires the following: 
\begin{equation} \label{eq:anti_comonotone_standard}
(T(\theta, \theta')-T(\varphi, \varphi')) \cdot \left( \frac{dQ_T}{dP}(\theta, \theta') - \frac{dQ_T}{dP}(\varphi, \varphi') \right) \leq 0.
\end{equation}
It is straightforward to verify that the first two lines of condition \eqref{eq:anti_comonotone_strong} are equivalent to condition \eqref{eq:anti_comonotone_standard}; hence, condition \eqref{eq:anti_comonotone_strong} implies condition \eqref{eq:anti_comonotone_standard}.

Then, Theorem \ref{thm:restriction} follows immediately from Proposition \ref{prop:anti_comonotone}.

\begin{proof} [Proof of Theorem \ref{thm:restriction} under Assumption \ref{assum:Q}]

Let $T(\theta, \theta') := t_1(\theta, \theta') + t_2(\theta', \theta)$. For a given $Q \in \mathcal{Q}$, define $Q_T$ as in Proposition \ref{prop:anti_comonotone} (i) (note that condition \eqref{eq:T_atomless} holds by Assumption \ref{assum:t} (ii)). Since $T$ increases in each argument by Assumption \ref{assum:t} (i), condition \eqref{eq:anti_comonotone_strong} implies that $dQ_T/dP$ decreases in each argument.\footnote{In contrast, the standard anti-comonotonicity condition \eqref{eq:anti_comonotone_standard} does not imply monotonicity of $dQ_T/dP$. For example, if $T$ is constant, then condition \eqref{eq:anti_comonotone_standard} is trivially satisfied.} Hence, $Q_T$ is a decreasing rearrangement of $Q$, i.e., $Q_T \in \mathcal{Q}^*$. Also, by Proposition \ref{prop:anti_comonotone} (ii), the expected revenue under $Q^*$ is lower than that under $Q$. Thus, the minimum expected revenue over $\mathcal{Q}$ equals that over $\mathcal{Q}^*$.

\end{proof}

\begin{figure} [t]
\centering
\includegraphics[scale=1]{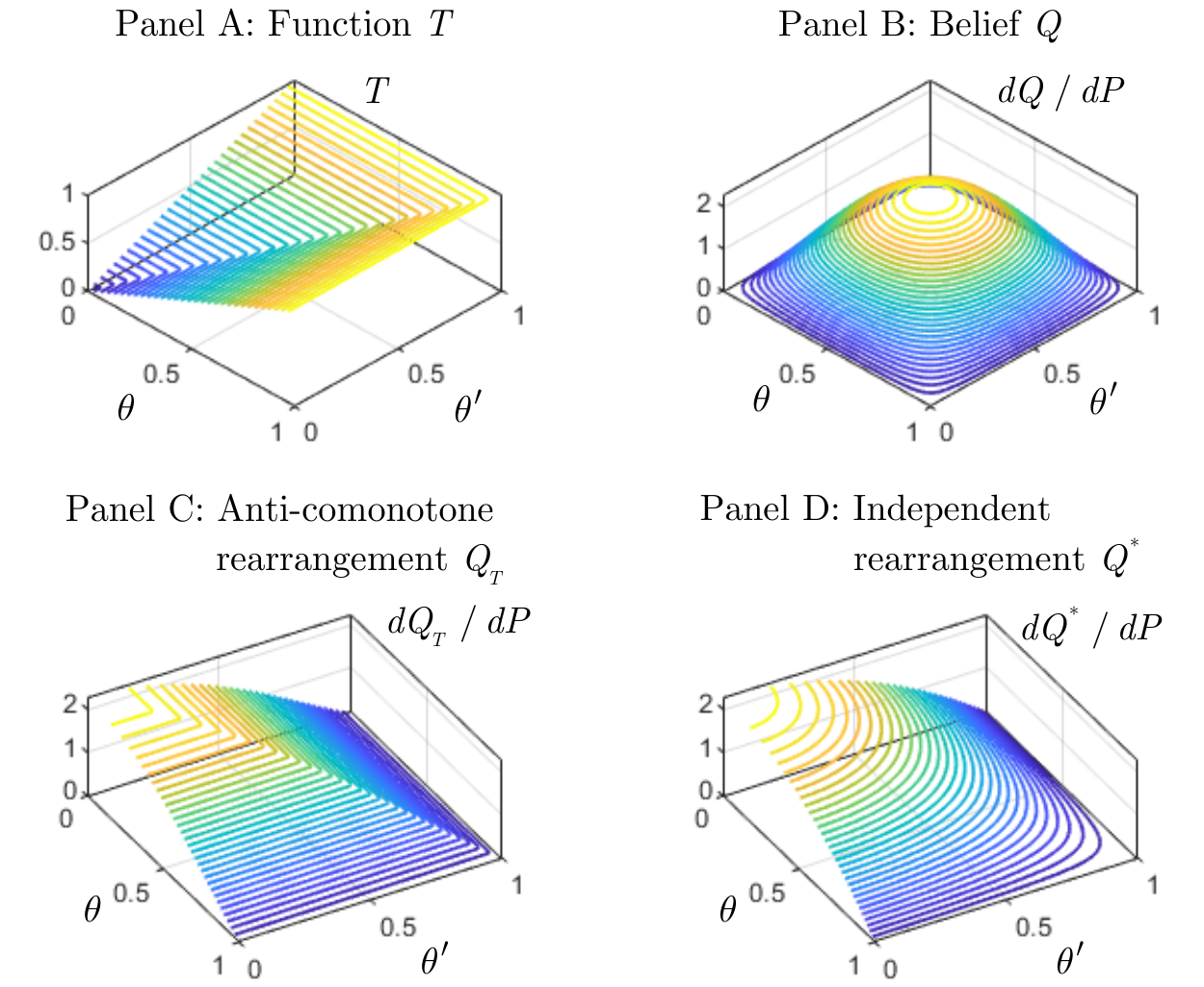}
\caption{\textbf{Rearrangements in Proposition \ref{prop:anti_comonotone} and \ref{prop:rearrangement_ineq_IID}.} Let $P$ be uniform over $\Theta^2$. Suppose $T: \Theta^2 \rightarrow \mathbb{R}_+$ is given as in Panel A ($T(\theta, \theta') := \max \{\theta, \theta' \}$), and $Q \in \Delta(\Theta^2, P)$ as in Panel B. Then, Panel C shows the anti-comonotone rearrangement $Q_T$ (Proposition \ref{prop:anti_comonotone}), and Panel D shows the independent and decreasing rearrangement $Q^*$ (Proposition \ref{prop:rearrangement_ineq_IID}).} \label{fig:rearrangement_ineq_intuition}
\end{figure}

\noindent \textit{Case B: $\mathcal{Q}$ satisfies Assumption \ref{assum:Q_indep_IID}.} In this case, Theorem \ref{thm:restriction} does not follow from Proposition \ref{prop:anti_comonotone}. This is because even if $Q$ is independent (or IID), the anti-comonotone rearrangement $Q_T$ in Proposition \ref{prop:anti_comonotone} is not necessarily independent (or IID). However, Proposition \ref{prop:rearrangement_ineq_IID} shows that given an independent (or IID) belief, we can construct a rearrangement satisfying statements similar to Proposition \ref{prop:anti_comonotone} while preserving the independence (or IID) property.

\begin{proposition} \label{prop:rearrangement_ineq_IID}
Let $T: \Theta^2 \rightarrow \mathbb{R}_+$ be increasing in each argument and $Q \in \Delta(\Theta^2, P)$.

\noindent (i) If $P$ and $Q$ are independent, there exists an independent and decreasing rearrangement $Q^*$ of $Q$. Furthermore, if $P$ and $Q$ are IID, then $Q^*$ is IID.

\noindent (ii) Moreover, the expectation of $T$ is lower under $Q^*$ than under $Q$:
\[
\iint_{\Theta^2} T(\theta, \theta') Q^*(d\theta, d\theta') \leq \iint_{\Theta^2} T(\theta, \theta') Q(d\theta, d\theta').
\]

\end{proposition}
\begin{proof}
See Appendix \ref{appen:rearrangement_ineq_IID}.
\end{proof}

Panel D of Figure \ref{fig:rearrangement_ineq_intuition} illustrates $Q^*$. Unlike the anti-comonotone rearrangement $Q_T$ in Proposition \ref{prop:anti_comonotone}, the upper contour sets of $dQ^*/dP$ do not coincide with the lower contour sets of $T$. However, the upper contour sets of $dQ^*/dP$ have greater intersections with the lower contour sets of $T$ than the upper contour sets of $dQ/dP$ have. This means that $Q^*$ assigns greater probabilities to low values of $T$ and smaller probabilities to high values of $T$ than $Q$ does. Hence, the expectation of $T$ is lower under $Q^*$ than under $Q$.

The proof of Theorem \ref{thm:restriction} under Assumption \ref{assum:Q_indep_IID} is similar to that under Assumption \ref{assum:Q}, and hence omitted.

\section{Comparison between commonly studied auctions} \label{sec:ranking_4}

In this section, assuming that the reference belief $P$ is IID, we apply Theorem \ref{thm:methodology} to compare the worst-case revenues of four commonly studied auctions: the first-price auction (I), second-price auction (II), all-pay auction (A) and (static) war of attrition (W). For simplicity, we assume no reserve price; however, the extension to reserve prices is straightforward.

Since the bidders are assumed to be ambiguity neutral, the equilibrium bidding strategies and transfer functions for the four auctions are given as follows \citep[see, e.g.,][]{Mil04}:
\[
\begin{array} {ll}
b^I(\theta) := \theta - \int_{0}^\theta \frac{F(z)}{F(\theta)} dz &\;\;\, t^I_i(\theta, \theta') := b^I(\theta) \cdot ( \mathbf{1}[\theta>\theta'] + \frac{1}{2} \mathbf{1}[\theta=\theta'] ) \\
b^{II}(\theta) := \theta &\;\;\, t^{II}_i(\theta, \theta') := b^{II}(\theta') \cdot ( \mathbf{1}[\theta>\theta'] + \frac{1}{2} \mathbf{1}[\theta=\theta'] ) \\
b^A(\theta) := \theta F(\theta) - \int_{0}^\theta F(z) dz  &\;\;\, t^A_i(\theta, \theta') := b^A(\theta) \\
b^W(\theta) := \int_{0}^\theta \frac{z f(z)}{1-F(z)} dz &\;\;\, t^W_i(\theta, \theta') := b^W(\theta') \mathbf{1}[\theta > \theta'] + b^W(\theta) \mathbf{1}[\theta \leq \theta'],
\end{array}
\]
where $F(z) := P \{ (\theta, \theta'): \theta \leq z, \theta' \in \Theta \}$ denotes the marginal cumulative distribution of $P$ and $f(z):=F'(z)$ denotes its marginal probability density.

Theorem \ref{thm:ranking}, the main result of this section, establishes the worst-case revenue rankings between the four auctions (Figure \ref{fig:summary}).

\begin{figure} [t]
\centering
\includegraphics[scale=1]{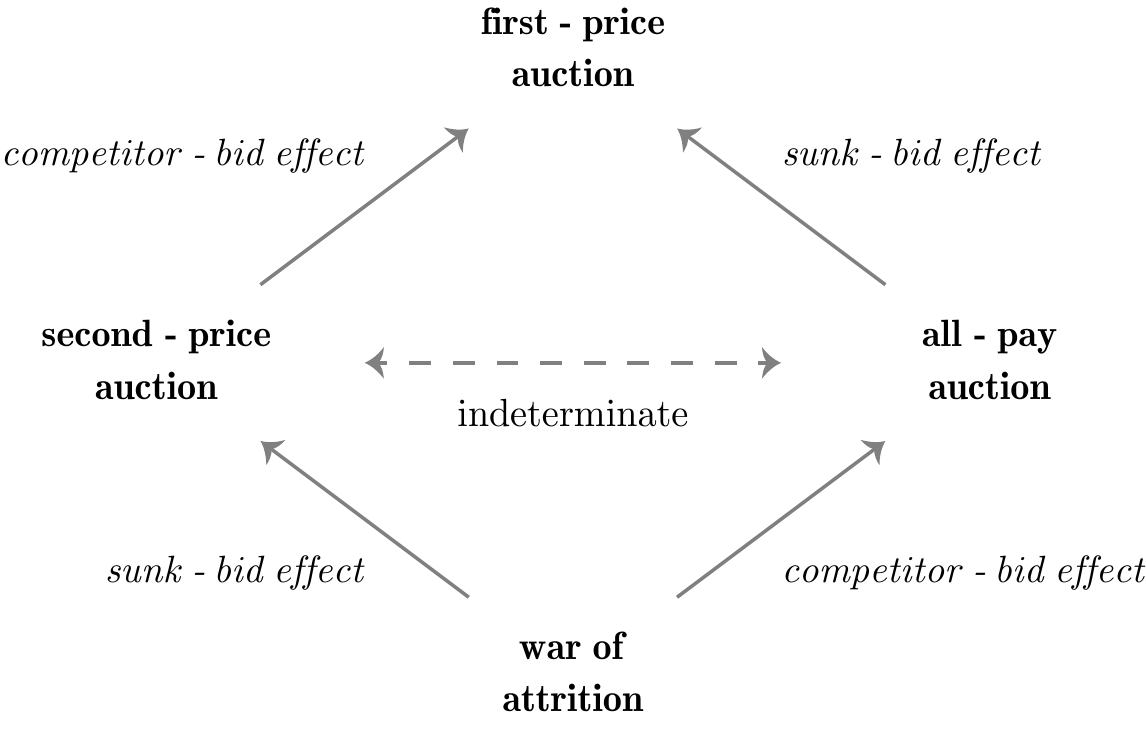}
\caption{\textbf{Theorem \ref{thm:ranking}.} Arrows indicate the direction in which the worst-case revenue increases.} \label{fig:summary}
\end{figure}

\begin{theorem} \label{thm:ranking}

Suppose $\mathcal{Q}$ satisfies Assumption \ref{assum:Q} or \ref{assum:Q_indep_IID}. If $P$ is IID and the bidders are ambiguity neutral, the following statements hold:

\noindent (i) $\mathcal{R}(t^I) \geq \mathcal{R}(t^{II})$. 

\noindent (ii) $\mathcal{R}(t^I) \geq \mathcal{R}(t^A)$. 

\noindent (iii) $\mathcal{R}(t^A) \geq \mathcal{R}(t^W)$.

\noindent (iv) Suppose that
\begin{equation} \label{eq:F_hazard}
\text{$\theta \times f(\theta) / [1-F(\theta)]$ increases in $\theta$.}
\end{equation}
Then, $\mathcal{R}(t^{II}) \geq \mathcal{R}(t^W)$. 

\end{theorem}
\begin{proof}
See Appendix \ref{appen:ranking}.
\end{proof}

\noindent Condition \eqref{eq:F_hazard} is a weak version of the usual assumption that the hazard rate $f/(1-F)$ is increasing. This condition guarantees that the equilibrium bidding strategies of the second-price auction and war of attrition, $b^{II}$ and $b^W$, intersect exactly once (except at the origin).\footnote{Condition \eqref{eq:F_hazard} can be weakened because Theorem \ref{thm:ranking} (iv) holds whenever $b^{II}$ and $b^W$ intersect exactly once (except at the origin).} For example, $F(\theta) \equiv \theta^\alpha$ satisfies condition \eqref{eq:F_hazard}, where $\alpha>0$.

Notably, the worst-case revenue rankings in Theorem \ref{thm:ranking} are opposite to the expected revenue rankings in the affiliated values setup \citep{Mil82, Kris97}. Section \ref{sec:linkage} discusses the relationship between the two in detail. Also, in the special case of the bounded likelihood ratio model (Example \ref{ex:Indep_IID} (b-IID)), \cite{Lo98} shows that the first-price auction outperforms the second-price auction. Theorem \ref{thm:ranking} (i) generalizes this result.\footnote{\cite{Lo98} also analyzes the case where both the seller and bidders have MMEU preferences, with the sets of priors given by the bounded likelihood ratio model. Likewise, this result is a special case of Corollary \ref{cor:bidder_aversion} (i) in Section \ref{subsec:bidder_aversion}.}

\begin{figure} [t]
\centering
\includegraphics[scale=1]{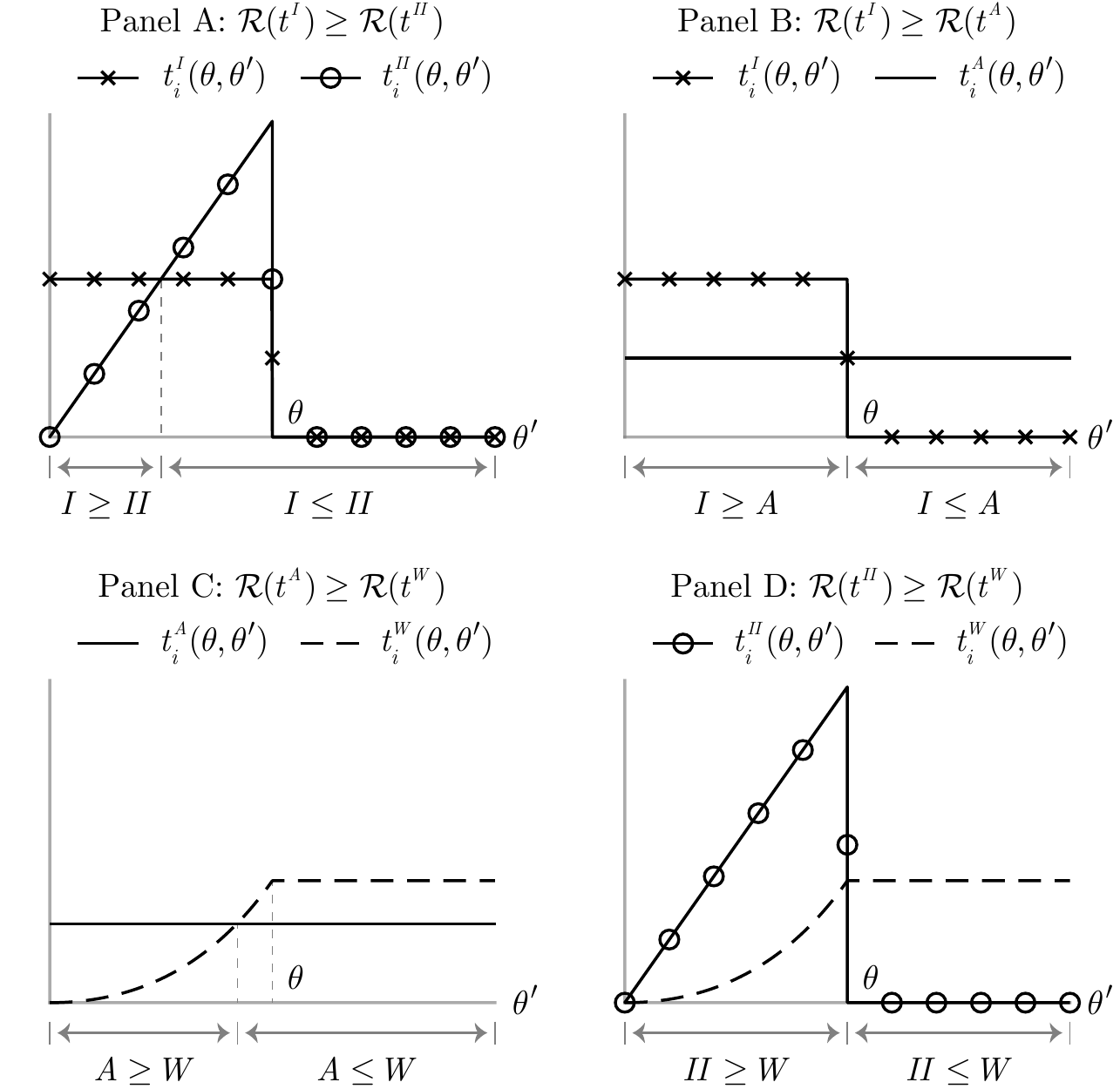}
\captionof{figure}{\textbf{Proof of Theorem \ref{thm:ranking}.} Each panel plots the transfer functions of a type $\theta$ bidder in two auctions. Horizontal axes represent the competitor's type $\theta'$.} \label{fig:payment_schedule_comparison}
\end{figure}

We now outline the proof of Theorem \ref{thm:ranking}. By Theorem \ref{thm:methodology}, to prove Theorem \ref{thm:ranking}, it is sufficient to verify that the pairs $(X, Y) = (I, II), (I, A), (A, W), (II, W)$ satisfy both WSCC and RRC. Figure \ref{fig:payment_schedule_comparison} illustrates that these pairs satisfy WSCC. Also, by the Revenue Equivalence Principle \citep{Myer81}, the four auctions yield the same interim expected revenues under $P$; hence, RRC holds as an equality. This establishes Theorem \ref{thm:ranking}.

Analogously to \cite{Kris97}, we identify two independent effects driving Theorem \ref{thm:ranking}. To explain this, recall from Theorem \ref{thm:methodology} that a negative (or positive) association between a bidder's payment and her competitor's type increases (or decreases) worst-case revenue. First, auctions in which a bidder pays the competitor's bid underperform auctions in which a bidder pays her own bid; we name this the \textit{competitor-bid effect} (Figure \ref{fig:summary}, arrows with upper-right directions). When a bidder pays the competitor's bid instead of her bid, a positive association between her payment and the competitor's type arises. According to Theorem \ref{thm:methodology}, this positive association decreases worst-case revenue. The competitor-bid effect explains why the second-price auction underperforms the first-price auction, and the war of attrition underperforms the all-pay auction.

Second, auctions in which bids are sunk underperform auctions in which payments are contingent on winning; we name this the \textit{sunk-bid effect} (Figure \ref{fig:summary}, arrows with upper-left directions). The logic is similar as in the previous paragraph: the fact that a bidder pays even when she loses---in which case the competitor's type is high---creates a positive association between her payment and the competitor's type. The sunk-bid effect explains why the all-pay auction underperforms the first-price auction, and the second-price auction underperforms the war of attrition.

\begin{figure} [t]
\centering
\includegraphics[scale=1]{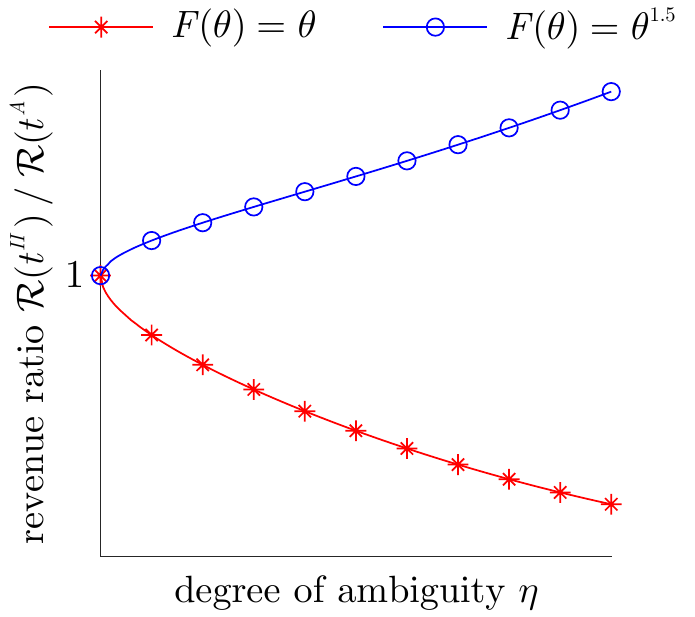}
\caption{\textbf{Indeterminacy between II and A.}  Let $\mathcal{Q}$ be the relative entropy neighborhood (Example \ref{ex:Q} (a)): $\mathcal{Q} := \{ Q \in \Delta(\Theta^2, P): \int \log (dQ/dP) dQ \leq \eta \}$. The figure plots the worst-case revenue ratio $\mathcal{R}(t^{II})/\mathcal{R}(t^{A})$, where the horizontal axis represents the degree of ambiguity $\eta$. The starred and circled lines represent the cases where $F(\theta)=\theta$ and $F(\theta)=\theta^{1.5}$. In the former, the second-price auction outperforms the all-pay auction; in the latter, the opposite holds.} \label{fig:SPA_APA_indeterminate}
\end{figure}

The ranking between the second-price and all-pay auctions is indeterminate, as shown in Figure \ref{fig:SPA_APA_indeterminate}. This is because whereas the competitor-bid effect makes the second-price auction inferior to the all-pay auction, the sunk-bid effect offsets this effect.

\section{Relation to the Linkage Principle} \label{sec:linkage}

As mentioned in Section \ref{sec:ranking_4}, the worst-case revenue rankings between the four auctions in Theorem \ref{thm:ranking} are opposite to the expected revenue rankings in the affiliated values setup \citep{Mil82, Kris97}. By investigating the relationship between Theorem \ref{thm:methodology} and the Linkage Principle, this section explains why the two results are opposite.

We first introduce some notation and terminology. Let $f: \Theta^2 \rightarrow \mathbb{R}_+$ be the probability density of $P$. Recall that $P$ is \textit{symmetric} if $f(\theta, \theta') = f(\theta', \theta)$, and $P$ is \textit{affiliated} if for all $\theta, \theta', \varphi, \varphi' \in \Theta$,
\[
f(\theta, \theta') f(\varphi, \varphi') \leq f(\max \{ \theta, \varphi \}, \max \{ \theta', \varphi' \}) f(\min \{ \theta, \varphi \}, \min \{ \theta', \varphi' \}).
\]
In this section, we focus on symmetric auctions in which the highest bidder wins. Consider an auction with a symmetric and increasing equilibrium. Let $e_i(\hat \theta, \theta)$ and $w_i(\hat \theta, \theta)$ be the unconditional expected payment and the expected payment conditional on winning when bidder $i$ with type $\theta$ reports $\hat \theta$:
\[
e_i(\hat \theta, \theta) := \int_\Theta t_i(\hat \theta, \theta') P(d\theta'|\theta), \quad \text{and} \quad w_i(\hat \theta, \theta) := \int_0^{\hat \theta} t_i(\hat \theta, \theta') \frac{P(d\theta'|\theta)}{P([0, \hat \theta]|\theta)}.
\]
Also, $\partial_2$ denotes the partial derivative with respect to the second argument.

Next, recall the Linkage Principle:

\begin{theorem} [Linkage Principle; \citealp{Kris02}, Ch. 7] \label{thm:linkage}

Assume $P$ is symmetric and affiliated. Let $X$ and $Y$ be auctions with symmetric and increasing equilibria satisfying $e_i^X(0, 0)=e_i^Y(0, 0) = 0$. Suppose that either of the following two conditions holds:

\noindent (i) \textbf{Linkage Condition-Version 1 (LC1).} For all $i$ and $\theta$,
\[
\partial_2 e^X_i (\theta, \theta) \leq \partial_2 e^Y_i (\theta, \theta).
\]

\noindent (ii) \textbf{Linkage Condition-Version 2 (LC2).} The loser pays nothing in both $X$ and $Y$. In addition, for all $i$ and $\theta$,
\[
\partial_2 w_i^X(\theta, \theta) \leq \partial_2 w_i^Y(\theta, \theta).
\]

\noindent Then,
\[
\int_\Theta t^X_i(\theta, \theta') P(d\theta'|\theta) \leq \int_\Theta t^Y_i(\theta, \theta') P(d\theta'|\theta).
\]

\end{theorem}

According to Theorem \ref{thm:methodology}, when $P$ is IID, WSCC implies that $X$ yields a higher worst-case revenue than $Y$ (recall that RRC automatically holds by the Revenue Equivalence Principle). On the contrary, according to the Linkage Principle, when $P$ is symmetric and affiliated, LCs imply that $Y$ yields a higher expected revenue than $X$. Proposition \ref{prop:linkage_4} below shows that between the four auctions studied in Section \ref{sec:ranking_4}, WSCC holds if and only if either of the two LCs holds. Thus, Theorem \ref{thm:methodology} and the Linkage Principle work in the opposite directions. This explains why the worst-case revenue rankings in Theorem \ref{thm:ranking} are opposite to the expected revenue rankings with affiliated values.

\begin{proposition} \label{prop:linkage_4}

For $X \neq Y \in \{ I, II, A, W \}$, the following conditions are equivalent:

\noindent (i) Under any IID $P$, $(X, Y)$ satisfies WSCC.

\noindent (ii) Under any symmetric and affiliated $P$ such that $X$ and $Y$ have symmetric and increasing equilibria,\footnote{The first-price and second-price auctions have equilibria whenever $P$ is symmetric and affiliated. \cite{Kris97} provide sufficient conditions on $P$ for equilibrium existence in the all-pay auction and war of attrition, omitted in our paper due to space limitation.} $(X, Y)$ satisfies either LC1 or LC2.

\end{proposition}
\begin{proof}

See Appendix \ref{appen:linkage_4}.

\end{proof}

\begin{figure}[t]
\centering
\includegraphics[scale=1]{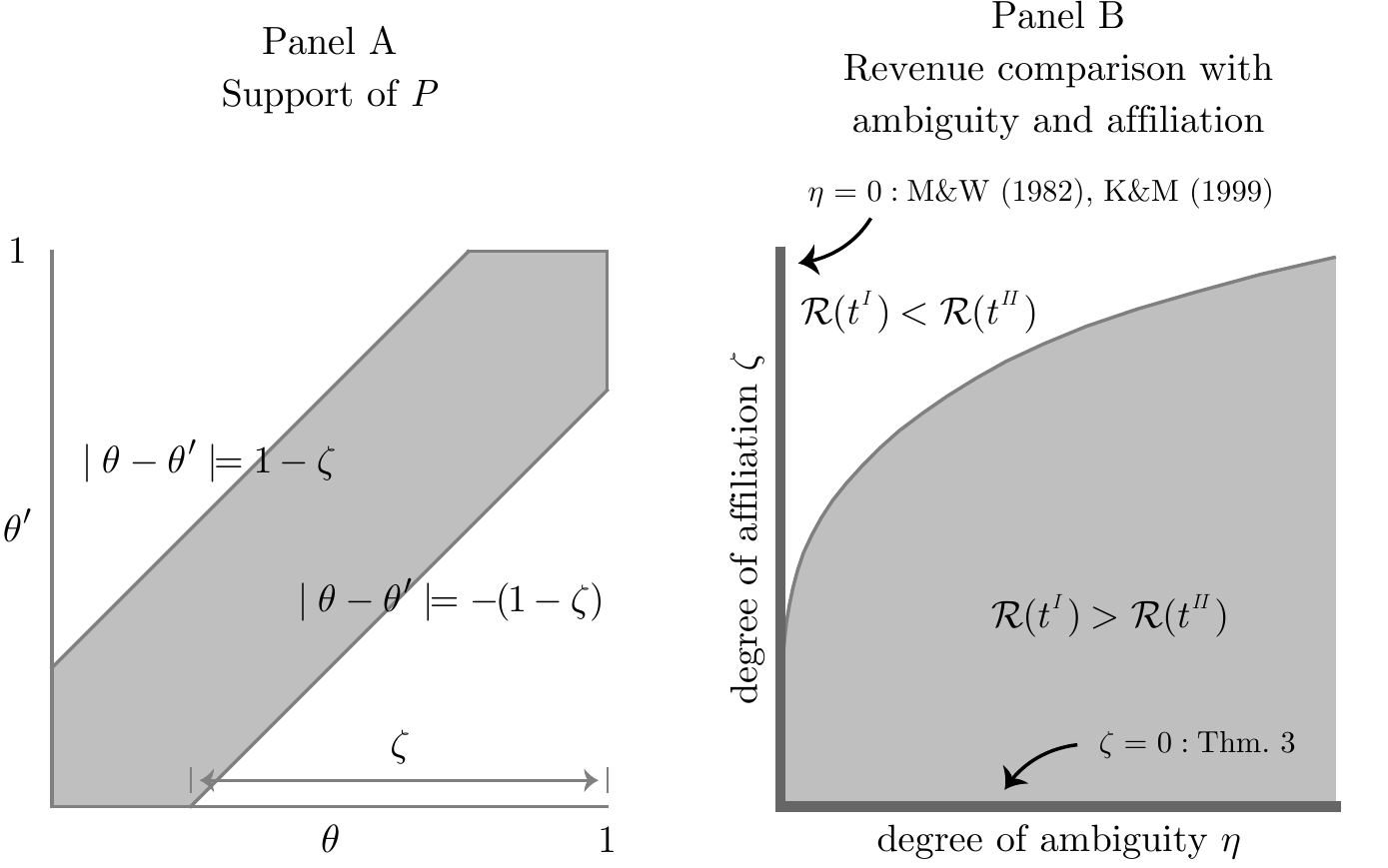}
\caption{\textbf{Ambiguity vs. affiliation: Indeterminacy in the presence of both.} \\
Let $\mathcal{Q}$ be the relative entropy neighborhood (Example \ref{ex:Q} (a)): $\mathcal{Q} := \{ Q \in \Delta(\Theta^2, P): \int \log (dQ/dP) dQ \leq \eta \}$. Also, let $P$ be uniform over $\{ (\theta, \theta') \in \Theta^2: |\theta-\theta'| \leq 1-\zeta \}$, illustrated in Panel A. If $\zeta=0$, types are independent; if $\zeta=1$, types are perfectly affiliated. The parameters $\eta \geq 0$ and $\zeta \in [0, 1]$ represent the degrees of ambiguity and affiliation, respectively.\\
Panel B compares the worst-case revenues of the first-price and second-price auctions for each $(\eta, \zeta)$. In the shaded region where ambiguity dominates affiliation, the ranking is the same as in Theorem \ref{thm:ranking}. By contrast, in the white region where affiliation dominates ambiguity, the ranking is the same as in \cite{Mil82} and \cite{Kris97}.} \label{fig:affiliation}
\end{figure}

Intuitively, WSCC means that a bidder's payment is more negatively associated with her competitor's type in $X$ than in $Y$. On the other hand, the standard interpretation of LCs is that a bidder's payment is more negatively associated with her own type in $X$ than in $Y$. However, a negative association between a bidder's payment and her competitor's type creates a negative association between her payment and her own type in the affiliated value setup. As a result, WSCC and LCs hold simultaneously.

A direct implication of Proposition \ref{prop:linkage_4} is that in the presence of both ambiguity and affiliation, the revenue rankings between the four auctions in Section \ref{sec:ranking_4} are indeterminate. When the effect of ambiguity dominates the effect of affiliation, the ranking is the same as in Theorem \ref{thm:ranking}; in the opposite case, the ranking is the same as in \cite{Mil82} and \cite{Kris97}. Figure \ref{fig:affiliation} illustrates this fact by comparing the first-price and second-price auctions. Comparisons between other pairs of auctions yield similar results.

\section{Extensions} \label{sec:extension}

This section presents two extensions: ambiguity averse bidders (Section \ref{subsec:bidder_aversion}) and ambiguity seeking seller (Section \ref{subsec:seeking}).

\subsection{Ambiguity averse bidders} \label{subsec:bidder_aversion}

Existing studies on auctions with ambiguity mostly focus on the implications of the bidders' ambiguity aversion (Table \ref{tab:extension}; see also Section \ref{subsec:literature}). This section partially extends our results to the setup in which both the seller and bidders exhibit ambiguity aversion. Specifically, assuming that the reference belief is IID, we show that the worst-case revenue comparison results between the four auctions remain unchanged, except for the case between the second-price auction and war of attrition.

\begin{table} [t]
\centering
{\small
\begin{tabular}{|c|c|c|c|} \hline
\diaghead{\theadfont Bidder \,\, Seller}
{Bidder}{Seller}&
\begin{tabular}{@{}c@{}} Ambiguity \\ neutral \end{tabular} & \begin{tabular}{@{}c@{}} Ambiguity \\ averse \end{tabular} & \begin{tabular}{@{}c@{}} Ambiguity \\ seeking \end{tabular} \\ \hline
\begin{tabular}{@{}c@{}} Ambiguity \\ neutral \end{tabular} & \begin{tabular}{@{}c@{}} \cite{Myer81} \\ \cite{Mil82} \\
\cite{Kris97} \end{tabular} & \begin{tabular}{@{}c@{}} \citet[Sec. 6]{Bo06} \\ Sections \ref{sec:methodology}-\ref{sec:linkage} \end{tabular} & Section \ref{subsec:seeking}  \\ \hline
 \begin{tabular}{@{}c@{}} Ambiguity \\ averse \end{tabular} & \begin{tabular}{@{}c@{}} \cite{Bo09} \\ \cite{Bodoh12} \\ \cite{Aus22} \\ \cite{Ghosh21} \\ \cite{Baik21} \end{tabular} & \begin{tabular}{@{}c@{}} \cite{Lo98} \\ \citet[Sec. 3]{Bo06} \\ Section \ref{subsec:bidder_aversion} \end{tabular} & - \\ \hline
\end{tabular}
}
\caption{\textbf{Comparison of the setups.} This table compares the setups studied in related works and in each section of this paper.} \label{tab:extension}
\end{table}

We represent a bidder's belief about the competitor's type by its cumulative distribution function $G: \Theta \rightarrow [0, 1]$. Also, denote a bidder's reference belief by $F(z) := P \{ (\theta, \theta'): \theta \leq z, \theta' \in \Theta \}$. Each bidder holds a set of priors $\mathcal{Q}^B$ (assumed to be weakly compact) that satisfies the following assumption:

\begin{assumption} \label{assum:Q_bidder}

\noindent (i) $\mathcal{Q}^B \subset \Delta(\Theta, F)$.

\noindent (ii) $\mathcal{Q}^B$ is rearrangement invariant with respect to $\Delta(\Theta, F)$.

\end{assumption}

Consider a symmetric sealed-bid auction in which a bidder wins the object with probability $x(b, b')$ and pays $\tau(b, b')$ when she bids $b$ and the competitor bids $b'$. A bidding strategy $b^*: \Theta \rightarrow \mathbb{R}_+$ is a symmetric equilibrium if
\[
b^*(\theta) \in \argmax_b \min_{G \in \mathcal{Q}^B} \int_\Theta \left[ \theta x(b, b^*(\theta')) - \tau(b, b^*(\theta')) \right] dG(\theta') \quad \text{for all $\theta$.}
\]
Then, the transfer function is given by $t_i(\theta, \theta') := \tau(b^*(\theta), b^*(\theta'))$.

Existing studies provide closed-form formulas for the equilibria of the first-price, second-price and all-pay auctions \citep{Lo98, Baik21}. Regarding the war of attrition, although an implicit characterization of equilibrium is available, a sufficient condition for the existence of equilibrium is unknown. As the investigation of this problem is beyond the scope of this paper, we simply assume the existence of equilibrium when necessary.

Now, we compare the worst-case revenues of the four auctions in Section \ref{sec:ranking_4}. Because Theorem \ref{thm:methodology} imposes no restrictions on the bidders' preferences, it is applicable to the current setup. Therefore, to show that the seller prefers auction $X$ to auction $Y$, it suffices to verify WSCC and RRC. Arguing as in Section \ref{sec:ranking_4}, it is straightforward to prove that the pairs $(X, Y) = (I, II), (I, A), (A, W)$ satisfy WSCC (provided that the war of attrition has an equilibrium). In addition, Proposition \ref{prop:bidder_aversion}, proven by \cite{Baik21}, shows that these pairs of auctions satisfy RRC:\footnote{\cites{Baik21} assumption on the bidders' sets of priors differs from ours. However, their proofs are valid as long as the following property holds: for any bounded measurable $\pi: \Theta \rightarrow \mathbb{R}$ and a $\sigma$-algebra $\mathcal{E}$,
\[
\min_{G \in \mathcal{Q}^B} \int_\Theta \mathbb{E}_F[\pi|\mathcal{E}] d\nu \geq \min_{G \in \mathcal{Q}^B} \int_\Theta \pi d\nu.
\]
\citet[Thm. 2]{Cer12} show that Assumption \ref{assum:Q_bidder} implies this property.}

\begin{proposition} [\citealp{Baik21}] \label{prop:bidder_aversion}

Suppose $P$ is IID and $\mathcal{Q}^B$ satisfies Assumption \ref{assum:Q_bidder}. Then, for all $i$ and $\theta$,

\noindent (i) $\int t^I_i(\theta, \theta') dF(\theta') \geq \int t^{II}_i(\theta, \theta') dF(\theta')$.

\noindent (ii) $\int t^I_i(\theta, \theta') dF(\theta') \geq \int t^A_i(\theta, \theta') dF(\theta')$.

\noindent (iii) If the war of attrition has an equilibrium, $\int t^A_i(\theta, \theta') dF(\theta') \geq \int t^W_i(\theta, \theta') dF(\theta')$.

\end{proposition}

As a result, we obtain Corollary \ref{cor:bidder_aversion}, which states that Theorem \ref{thm:ranking} (i)-(iii) remain valid when the bidders are ambiguity averse. 

\begin{corollary} \label{cor:bidder_aversion}
Suppose $P$ is IID, $\mathcal{Q}$ satisfies Assumption \ref{assum:Q} or \ref{assum:Q_indep_IID}, and $\mathcal{Q}^B$ satisfies Assumption \ref{assum:Q_bidder}. Then, 

\noindent (i) $\mathcal{R}(t^I) \geq \mathcal{R}(t^{II})$.

\noindent (ii) $\mathcal{R}(t^I) \geq \mathcal{R}(t^A)$. 

\noindent (iii) If the war of attrition has an equilibrium, $\mathcal{R}(t^A) \geq \mathcal{R}(t^W)$.
\end{corollary}

\noindent The primary difficulty with extending Theorem \ref{thm:ranking} (iv), which compares the second-price auction and war of attrition, lies in the complexity of the equilibrium characterization of the war of attrition.

\begin{remark} \normalfont 

\cite{Aus22} analyze the Dutch auction with ambiguity averse bidders. They find that due to dynamic inconsistency, the strategic equivalence between the Dutch and first-price auctions breaks down, and the equilibrium bidding strategy of the Dutch auction is higher than that of the first-price auction. This result, combined with Corollary \ref{cor:bidder_aversion}, implies that when both the seller and bidders exhibit ambiguity aversion, the Dutch auction outperforms the four static auctions studied in this section.

\end{remark}

\subsection{Ambiguity seeking seller} \label{subsec:seeking}

Experimental evidence shows that there is substantial heterogeneity in individuals' attitudes toward ambiguity, and some individuals are ambiguity seeking \citep{Ahn14, Chan22}. This section studies the setup in which the seller displays an ambiguity seeking preference. That is, she evaluates an auction by the \textit{best-case revenue} $\mathcal{R}^{\max}(t)$, defined as
\[
    \mathcal{R}^{\max}(t) := \max_{Q \in \mathcal{Q}} \iint_{\Theta^2} [t_1(\theta, \theta')+t_2(\theta', \theta)] Q(d\theta, d\theta').
\]

Proposition \ref{prop:ambiguity_seeking} states that if auctions $X$ and $Y$ satisfy the opposite condition to WSCC---named the \textit{Negative Weak Single-Crossing Condition (NWSCC)}---and RRC, then the ambiguity seeking seller prefers $X$ to $Y$. The name NWSCC is derived from the fact that it requires the negatives of transfer functions to satisfy WSCC: i.e., $t^X$ and $t^Y$ satisfy NWSCC if and only if $-t^X$ and $-t^Y$ satisfy WSCC.

\begin{proposition} \label{prop:ambiguity_seeking}

Suppose $\mathcal{Q}$ satisfies Assumption \ref{assum:Q} or \ref{assum:Q_indep_IID}. Let $X$ and $Y$ be auctions satisfying Assumption \ref{assum:t}. Consider the following condition:

\noindent \textbf{Negative Weak Single-Crossing Condition (NWSCC).} For all $i$ and $\theta$, there exists a threshold $\hat \theta \in [0, 1]$ such that
\[
\theta' < \hat \theta \implies t^X_i(\theta, \theta') \leq t^Y_i(\theta, \theta'), \quad \text{and} \quad \theta'>\hat \theta \implies t^X_i(\theta, \theta') \geq t^Y_i(\theta, \theta').
\]

\noindent If $(X, Y)$ satisfies NWSCC and RRC, then $\mathcal{R}^{\max}(t^X) \geq \mathcal{R}^{\max}(t^Y)$.

\end{proposition}

As an immediate consequence, revenue comparisons between the four auctions yield opposite results to the ambiguity aversion case: when $P$ is IID, (i) the war of attrition outperforms the second-price and all-pay auctions, and (ii) the second-price and all-pay auctions outperform the first-price auction. In other words, the best-case revenue rankings between the four auctions reproduce the expected revenue rankings in the affiliated values setup \citep{Mil82, Kris97}.

Because the two rankings are identical, unlike the ambiguity aversion case in Section \ref{sec:linkage}, the best-case revenue rankings extend to the case in which $P$ is symmetric and affiliated. Specifically, using the equilibrium bidding strategies in the affiliated values setup (\citealp{Mil82}, Thm. 6 and 14; \citealp{Kris97}, Thm. 1-2), it can be shown that the pairs $(X, Y) = (II, I), (A, I), (W, A), (W, II)$ satisfy NWSCC. In addition, the proofs of the expected revenue rankings in the affiliated values setup (\citealp{Mil82}, Thm. 15; \citealp{Kris97}, Thm. 3-5) show that the same rankings also hold for interim expected revenues; this implies that the above pairs of auctions satisfy RRC. Thus, by Proposition \ref{prop:ambiguity_seeking}, the best-case revenue rankings remain valid when $P$ is symmetric and affiliated.

\section{Discussion} \label{sec:discussion}

\subsection{Related literature} \label{subsec:literature}

\noindent \textit{Auctions with ambiguity.} This paper is most closely related to the literature on auctions with ambiguity. While existing works mainly focus on the bidders' ambiguity aversion \citep{Bo09, Bodoh12, Ghosh21, Aus22}, \citet[Sec. 6]{Bo06} show that when the seller is ambiguity averse and bidders are ambiguity neutral, the optimal mechanism is a \textit{seller-full-insurance auction} where the total transfer is constant in the type profile. \citet[Sec. 3]{Bo06} also show that when the seller and bidders are both ambiguity averse but the seller is less averse than the bidders, the optimal mechanism is a \textit{bidder-full-insurance auction} where a bidder's payoff is constant with respect to the competitor's type. However, because these two mechanisms depend on the bidders' beliefs, they are difficult to implement in practice and hence rarely used in reality \citep{Wil87}. We complement their results by comparing easily implementable auctions. Also, under a specific parametrization of the set of priors (Example \ref{ex:Indep_IID} (b-IID)), \cite{Lo98} compares the first-price and second-price auctions. As mentioned in Section \ref{sec:ranking_4}, our paper includes this result as a special case.

\medskip

\noindent \textit{Robust auction design.} This paper is also related to the robust auction design literature \citep{Berg17, Berg19, Brooks21, Che22, Suz22, He22}. Our assumption on the seller's ambiguity set differs from this literature. Existing works consider the minimum expected revenues over (i) all information structures between valuations and signals with a given valuation distribution \citep{Berg17, Berg19, Brooks21}, (ii) all valuation distributions satisfying certain moment conditions \citep{Che22, Suz22}, or (iii) all correlation structures between valuations with a given marginal valuation distribution \citep{He22}. By contrast, our set of priors consists of beliefs that are close to the reference belief, the so-called \textit{discrepancy-based} model \citep{Rah19}. Despite its popularity in other strands of the literature---e.g., the macroeconomics literature on model misspecification \citep{Han01, Han08} and the operations research literature on robust optimization \citep{Ben13}---the discrepancy-based model has been less frequently used in the literature on auctions where the seller has limited information about the valuation distribution. Our study shows that interesting revenue comparison results---those opposite to the Linkage Principle results---arise for discrepancy-based sets of priors.

\subsection{Conclusion}

This paper studies the revenue comparison problem of auctions when the seller has an MMEU preference. Assuming rearrangement invariance of the set of priors, we develop a methodology to compare the worst-case revenues. As an application, we compare the worst-case revenues of four commonly studied auctions: the first-price, second-price, all-pay auctions and war of attrition. Our methodology yields results opposite to those of the Linkage Principle.

Although this paper focuses on the four auctions, our methodology applies to a broader range of mechanisms. For instance, \cite{Sie10} studies a mechanism in which the winner pays her bid and the loser pays a fixed fraction of her bid, called a \textit{simple contest}. This mechanism can be regarded as a convex combination of the first-price and all-pay auctions. Applying Theorem \ref{thm:methodology}, it can be shown that the worst-case revenue of a simple contest decreases in the fraction of the bid paid by the loser. In other words, the closer a simple contest is to the first-price auction (equivalently, the farther it is from the all-pay auction), the higher the worst-case revenue it generates. Similar conclusions hold for the convex combinations of the other auction pairs studied in Section \ref{sec:ranking_4}.

Following most of the literature on auctions with ambiguity, our paper assumes that the seller has an MMEU preference. However, our results carry over to the setup in which the seller has an \textit{uncertainty averse preference}, a generalization of the MMEU preference axiomatized by \cite{Cer11}. Under rearrangement invariance assumptions analogous to Assumptions \ref{assum:Q}-\ref{assum:Q_indep_IID} \citep[see][Sec. 4.1]{Cer11}, it is straightforward to extend Theorems \ref{thm:restriction}-\ref{thm:ranking}. In particular, uncertainty averse preferences include \textit{divergence preferences} as a special case \citep{Macc06}, represented by
\[
\mathcal{R}^\text{div}(t) := \min_{Q \in \mathcal{Q}} \left[ \iint_{\Theta^2} [t_1(\theta, \theta')+t_2(\theta', \theta)] Q(d\theta, d\theta') + \frac{1}{\eta} D(Q||P) \right],
\]
where $D$ is defined in Example \ref{ex:Q} (a) and $\eta$ represents the degree of ambiguity. This preference, along with the MMEU preference, is one of the most popular models in the robustness literature \citep{Han01, Han08}.

\bigskip

\bigskip

\appendix

\noindent\textbf{\large Appendix}

\section{Rearrangement inequalities}

Throughout this section, we repeatedly use the following well-known fact in probability theory. Suppose that $\mu$ is an atomless probability measure on $\mathbb{R}_+$, and $G: \mathbb{R}_+ \rightarrow [0, 1]$ is the cumulative distribution of $\mu$:
\[
\text{for all $c \in \mathbb{R}_+$,} \quad G(c) := \mu \{ z \in \mathbb{R}_+: z \leq c \}.
\]
Then,
\begin{equation} \label{eq:cdf_dist_unif}
\text{for all $c \in [0, 1]$,} \quad \mu \{ z \in \mathbb{R}_+: G(z) \leq c \} = c.
\end{equation}
In other words, if $z$ is distributed according to $\mu$, then $G(z)$ is uniformly distributed over $[0, 1]$.

Also, for an increasing function $H: \mathbb{R}_+ \rightarrow [0, 1]$ such that $H(0) = 0$ and $\lim_{c \rightarrow \infty} H(c) = 1$, define its \textit{right-continuous inverse} $H^{-1}: [0, 1] \rightarrow \mathbb{R}_+$ as follows \citep[see, e.g.,][Sec. A.3]{Follmer16}:
\[
\text{for all $c \in [0, 1]$,} \quad H^{-1}(c) := \sup \{ z \in \mathbb{R}_+: H(z) \leq c \}.
\]

\subsection{Proof of Proposition \ref{prop:anti_comonotone}} \label{appen:anti_comonotone}

\noindent \textbf{(i)} Let $F_T, F_{dQ/dP}: \mathbb{R}_+ \rightarrow [0, 1]$ be the cumulative distributions of $T$ and $dQ/dP$:
\begin{align}
    \text{for all $c \in \mathbb{R}_+$,} \quad &F_T(c) := P \{ (\theta, \theta'): T(\theta, \theta') \leq c \} \notag \\
    &F_{dQ/dP}(c) := P \{ (\theta, \theta'): \frac{dQ}{dP}(\theta, \theta') \leq c \}. \label{eq:def_F_dQ_dP}
\end{align}
Applying equation \eqref{eq:cdf_dist_unif} with $\mu(E) := P \{ (\theta, \theta'): T(\theta, \theta') \in E \}$ and $G := F_T$ (note that $\mu$ is atomless by condition \eqref{eq:T_atomless}), we have
\begin{align} \label{eq:1_FT_T_unif}
&\text{for all $c \in [0, 1]$}, \quad \mu \{ z \in \mathbb{R}_+: F_T(z) \leq c \} = c \notag \\
\implies \quad &\text{for all $c \in [0, 1]$}, \quad P \{ (\theta, \theta'): F_T(T(\theta, \theta')) \leq c \} = c \notag \\
\implies \quad &\text{for all $c \in [0, 1]$}, \quad P \{ (\theta, \theta'): 1-F_T(T(\theta, \theta')) \leq c \} = c.
\end{align}
That is, if $(\theta, \theta')$ is distributed according to $P$, then $1-F_T(T(\theta, \theta'))$ is uniformly distributed over $[0, 1]$.

Now, define $Q^* \in \Delta(\Theta^2, P)$ as
\begin{equation} \label{eq:anti_comonotone_construction}
\frac{dQ^*}{dP}(\theta, \theta') := F_{dQ/dP}^{-1} \Big( 1-F_T(T(\theta, \theta')) \Big).
\end{equation}
Also, define $F_{dQ^*/dP}$ analogously as in equation \eqref{eq:def_F_dQ_dP}. Then, by Lemma A.23 of \cite{Follmer16}, equations \eqref{eq:1_FT_T_unif}-\eqref{eq:anti_comonotone_construction} imply $F_{dQ^*/dP} = F_{dQ/dP}$. This means that $Q^*$ is a rearrangement of $Q$. In addition, equation \eqref{eq:anti_comonotone_construction} shows that $dQ^*/dP$ is a decreasing function of $T$. This implies condition \eqref{eq:anti_comonotone_strong}. $\square$

\medskip

\noindent \textbf{(ii)} By Theorem A.28 of \cite{Follmer16},
\begin{align*}
\iint_{\Theta^2} T dQ^* &= \iint_{\Theta^2} T \frac{dQ^*}{dP} dP = \int_0^1 F_T^{-1}(1-z) F_{dQ^*/dP}^{-1}(z) dz \\
\iint_{\Theta^2} T dQ &= \iint_{\Theta^2} T \frac{dQ}{dP} dP \geq \int_0^1 F_T^{-1}(1-z) F_{dQ/dP}^{-1}(z) dz.
\end{align*}
Since $F_{dQ^*/dP} = F_{dQ/dP}$, it follows that $\iint_{\Theta^2} T dQ^* \leq \iint_{\Theta^2} T dQ$. $\square$

\subsection{Proof of Proposition \ref{prop:rearrangement_ineq_IID}} \label{appen:rearrangement_ineq_IID}

\begin{proof} [Proof of Proposition \ref{prop:rearrangement_ineq_IID} (i)]

Let $i \in \{1, 2\}$. It is evident from Appendix \ref{appen:anti_comonotone} that Proposition \ref{prop:anti_comonotone} (i) holds even if we replace $(\Theta^2, P)$ with $(\Theta, P_i)$. Applying this result with $T: \Theta \rightarrow \mathbb{R}_+$ defined as $T(\theta) \equiv \theta$, there exists a rearrangement $Q_i^* \in \Delta(\Theta, P_i)$ of $Q_i$ satisfying condition \eqref{eq:anti_comonotone_strong} (more precisely, its one-dimensional version). Since $T$ is increasing, condition \eqref{eq:anti_comonotone_strong} implies that $dQ_i^*/dP$ is decreasing. Now, let $Q^* := Q^*_{1} \times Q^*_{2}$. Then, $Q^*$ is an independent and decreasing rearrangement of $Q$. Also, by construction, if $P$ and $Q$ are IID, then $Q^*$ is IID.

\end{proof}

To prove Proposition \ref{prop:rearrangement_ineq_IID} (ii), as in the proofs of classical rearrangement inequalities \citep[][Ch. 3]{Lieb01}, we first consider the simplest case where $T$, $dQ/dP$ and $dQ^*/dP$ are indicator functions:
\begin{equation} \label{eq:simple_T_Q_Q*}
    T := \mathbf{1}_U, \quad \frac{dQ}{dP} := \frac{\mathbf{1}_{A_1 \times A_2}}{P(A_1 \times A_2)}, \quad \text{and} \quad \frac{dQ^*}{dP} := \frac{\mathbf{1}_{A_1^*\times A_2^*}}{P(A_1^*\times A_2^*)},
\end{equation}
where $U \subset \Theta^2$ and $A_i, A_i^* \subset \Theta$. Then, Proposition \ref{prop:rearrangement_ineq_IID} (ii) reduces to Lemma \ref{lem:rearrangement_ineq_simple}:

\begin{lemma} \label{lem:rearrangement_ineq_simple}

Assume $P$ is independent. Suppose that $U \subset \Theta^2$ is an event such that
\begin{equation} \label{eq:U_upper_set}
\theta_L \leq \theta_H, \theta_L' \leq \theta_H' \text{ and } (\theta_L, \theta'_L) \in U \implies (\theta_H, \theta_H') \in U.
\end{equation}
Let $A_1, A_2 \subset \Theta$ be events with $P(A_1 \times A_2)>0$. Also, let $A_1^*, A_2^* \subset \Theta$ be intervals with left endpoint $0$ satisfying $P_1(A_1^*)=P_1(A_1)$ and $P_2(A_2^*)=P_2(A_2)$. Then,
\begin{equation} \label{eq:rearrangement_ineq_simple}
P(U \cap (A_1^* \times A_2^*)) \leq P(U \cap (A_1 \times A_2)).\footnote{Note that if equation \eqref{eq:simple_T_Q_Q*} holds, then
\[
\iint_{\Theta^2} T dQ = \frac{P(U \cap (A_1 \times A_2))}{P(A_1 \times A_2)} \quad \text{and} \quad \iint_{\Theta^2} T dQ^* = \frac{P(U \cap (A_1^* \times A_2^*))}{P(A_1^* \times A_2^*)}.
\]
Since $P(A_1^* \times A_2^*) = P(A_1 \times A_2)$, inequality $\iint_{\Theta^2} T dQ^* \leq \iint_{\Theta^2} T dQ$ reduces to inequality \eqref{eq:rearrangement_ineq_simple}.}
\end{equation}

\end{lemma}

\noindent Panel A of Figure \ref{fig:rearrangement_ineq_simple} illustrates Lemma \ref{lem:rearrangement_ineq_simple}. As we show later, once Lemma \ref{lem:rearrangement_ineq_simple} is established, Proposition \ref{prop:rearrangement_ineq_IID} (ii) follows by a standard argument.

\begin{figure} [t]
\centering
\includegraphics[scale=1]{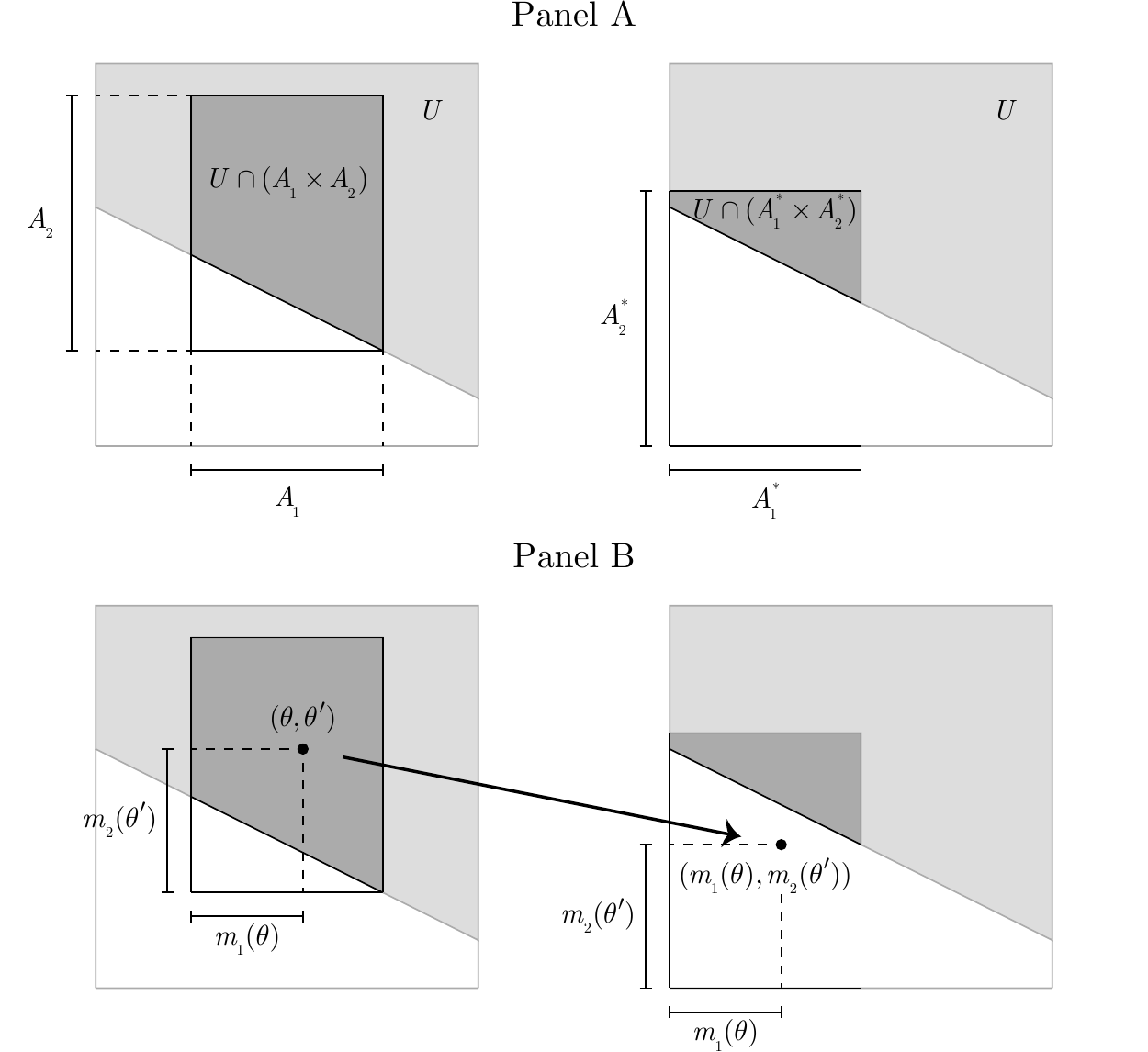}
\caption{\textbf{Lemma \ref{lem:rearrangement_ineq_simple}.} Let $P$ be uniform over $\Theta^2$. By assumption, $A_i$ and $A_i^*$ have the same total length. Panel A illustrates inequality \eqref{eq:rearrangement_ineq_simple}. The intersection of $A_1^* \times A_2^*$ with $U$ has a smaller area than that of $A_1 \times A_2$. Panel B illustrates Step 2 of the proof. Intuitively, the position of $(\theta, \theta')$ relative to $A_1 \times A_2$ equals that of $(m_1(\theta), m_2(\theta'))$ relative to $A_1^* \times A_2^*$. Hence, if $(\theta, \theta')$ drawn according to $Q$ (uniform over the left rectangle), then $(m_1(\theta), m_2(\theta'))$ is drawn according to $Q^*$ (uniform over the right rectangle).} \label{fig:rearrangement_ineq_simple}
\end{figure}

\begin{proof}[Proof of Lemma \ref{lem:rearrangement_ineq_simple}]

We proceed in three steps.

\noindent \textbf{Step 1.} \textit{Without loss of generality, we can assume $P$ is uniform over $\Theta$.}

Let $\lambda$ be the uniform probability measure on $\Theta$. Also, for $i \in \{1, 2\}$, denote the cumulative distribution of $P_i$ as $F_i: \Theta \rightarrow [0, 1]$. Define
\[
\begin{array} {ll}
\widehat U := \{ (F_1(\theta), F_2(\theta')): (\theta, \theta') \in U \} & \\
\widehat A_1 := \{ F_1(\theta): \theta \in A_1 \} & \quad \widehat A^*_1 := \{ F_1(\theta): \theta \in A_1^* \} \\
\widehat A_2 := \{ F_2(\theta'): \theta' \in A_2 \} & \quad \widehat A_2^* := \{ F_2(\theta'): \theta' \in A^*_2 \}.
\end{array}
\]
By equation \eqref{eq:cdf_dist_unif}, if $\theta$ is distributed according to $F_i$, then $F_i(\theta)$ is distributed according to $\lambda$. Using this fact, it is straightforward to verify the following:

(i) $\widehat U$ satisfies condition \eqref{eq:U_upper_set}.

(ii) $\widehat A_i^*$ is an interval with left endpoint $0$ satisfying $\lambda(\widehat A_i^*) = \lambda(\widehat A_i)$.

(iii) Inequality \eqref{eq:rearrangement_ineq_simple} is equivalent to 
\[
(\lambda \times \lambda) (\widehat U \cap (\widehat A_1^* \times \widehat A_2^*)) \leq (\lambda \times \lambda) (\widehat U \cap (\widehat A_1 \times \widehat A_2)).
\]
Hence, by replacing $U$, $A_1$, $A_2$ and $P$ with $\widehat U$, $\widehat A_1$, $\widehat A_2$ and $\lambda \times \lambda$, respectively, we can always assume $P = \lambda \times \lambda$. $\blacksquare$

\medskip

\noindent \textbf{Step 2.} \textit{For $i \in \{1, 2\}$, define $m_i: \Theta \rightarrow A_i^*$ as
\begin{equation} \label{eq:def_m}
    m_i(\theta) := P_i([0, \theta] \cap A_i).
\end{equation}
Also, define $Q, Q^* \in \Delta(\Theta^2, P)$ as in equation \eqref{eq:simple_T_Q_Q*}. Then, for all event $E \subset \Theta^2$, 
\begin{equation} \label{eq:Q_image_Q*}
Q \{ (\theta, \theta'): (m_1(\theta), m_2(\theta')) \in E \} = Q^*(E).
\end{equation}
}

Panel B of Figure \ref{fig:rearrangement_ineq_simple} illustrates Step 2. First, since $P_i$ is assumed to be uniform over $\Theta$ (Step 1), by definition \eqref{eq:simple_T_Q_Q*},
\begin{equation} \label{eq:rearrangement_ineq_simple_step2_Q*}
\text{for $c \in [0, P_i(A_i)]$,} \quad Q_i^*([0, c]) = \frac{P_i([0, c] \cap A_i^*)}{P_i(A_i^*)} = \frac{P_i([0, c])}{P_i(A_i^*)} = \frac{c}{P_i(A_i)}.
\end{equation}
Next, by definitions \eqref{eq:simple_T_Q_Q*} and \eqref{eq:def_m}, $m_i/P_i(A_i)$ is the cumulative distribution of $Q_i$. Hence, by equation \eqref{eq:cdf_dist_unif},
\begin{equation} \label{eq:rearrangement_ineq_simple_step2_Qm}
    \text{for $c \in [0, P_i(A_i)]$,} \quad Q_i \{ \theta: \frac{m_i(\theta)}{P_i(A_i)} \leq \frac{c}{P_i(A_i)} \} = \frac{c}{P_i(A_i)}.
\end{equation}
Then, equations \eqref{eq:rearrangement_ineq_simple_step2_Q*}-\eqref{eq:rearrangement_ineq_simple_step2_Qm} imply
\[
\text{for $c \in [0, P_i(A_i)]$,} \quad Q_i \{ \theta: m_i(\theta) \in [0, c] \} = Q_i^*([0, c]).
\]
If two measures coincide for sets of the form $[0, c]$, they must be equal. Hence,
\[
\text{for all event $E \subset \Theta$,} \quad Q_i \{ \theta: m_i(\theta) \in E \} = Q_i^*(E).
\]
By independence, equation \eqref{eq:Q_image_Q*} holds. $\blacksquare$

\medskip

\noindent \textbf{Step 3.} \textit{The desired inequality \eqref{eq:rearrangement_ineq_simple} holds.}

First, by definition \eqref{eq:simple_T_Q_Q*},
\begin{align} \label{eq:P_U_A1*_A2*}
    P(U \cap (A_1^* \times A_2^*)) &= P(A_1^* \times A_2^*) \iint_{\Theta^2} \mathbf{1}_U(x, y) \frac{\mathbf{1}_{A_1^* \times A_2^*}(x, y)}{P(A_1^* \times A_2^*)} P(dx, dy) \notag \\
    &= P(A_1^* \times A_2^*) \iint_{\Theta^2} \mathbf{1}_U(x, y) Q^*(dx, dy).
\end{align}
Next, Step 2 shows that if $(\theta, \theta')$ is distributed according to $Q$, then $(x, y) := (m_1(\theta), m_2(\theta'))$ is distributed according to $Q^*$.\footnote{More precisely, $Q^*$ is the image measure of $Q$ induced by $(\theta, \theta') \mapsto (m_1(\theta), m_2(\theta'))$.} By the change of variables formula for Lebesgue integration \citep[][Thm. 7 of Sec. II.6]{Shi96},
\begin{equation}
    \iint_{\Theta^2} \mathbf{1}_U(x, y) Q^*(dx, dy) = \iint_{\Theta^2}  \mathbf{1}_U(m_1(\theta), m_2(\theta')) Q(d\theta, d\theta').
\end{equation}
Now, by definition, $m_i(\theta) \leq \theta$. By property \eqref{eq:U_upper_set}, whenever $(m_1(\theta), m_2(\theta')) \in U$, we have $(\theta, \theta') \in U$. It follows that $\mathbf{1}_U(m_1(\theta), m_2(\theta')) \leq \mathbf{1}_U(\theta, \theta')$. Hence
\begin{equation}
     \iint_{\Theta^2}  \mathbf{1}_U(m_1(\theta), m_2(\theta')) Q(d\theta, d\theta') \leq \iint_{\Theta^2}  \mathbf{1}_U(\theta, \theta') Q(d\theta, d\theta').
\end{equation}
Finally, by the same argument as in equation \eqref{eq:P_U_A1*_A2*},
\begin{equation} \label{eq:P_U_A1_A2}
    P(U \cap (A_1 \times A_2)) = P(A_1 \times A_2) \iint_{\Theta^2} \mathbf{1}_U(\theta, \theta') Q(d\theta, d\theta').
\end{equation}
Since $P(A_1^* \times A_2^*) = P(A_1 \times A_2)$, \eqref{eq:P_U_A1*_A2*}-\eqref{eq:P_U_A1_A2} yields inequality \eqref{eq:rearrangement_ineq_simple}.

\end{proof}

\begin{proof} [Proof of Proposition \ref{prop:rearrangement_ineq_IID} (ii)]

By the Layer Cake Representation \citep[see, e.g.,][Sec. 1.13 and Sec. 3.4]{Lieb01} and Fubini's theorem,
\begin{align*}
&\iint_{\Theta^2} T(\theta, \theta') Q(d\theta, d\theta') = \iint_{\Theta^2} T(\theta, \theta') \frac{dQ_1}{dP_1}(\theta) \frac{dQ_2}{dP_2}(\theta') P(d\theta, d\theta') \\
&= \iint_{\Theta^2} \left[ \int_{\mathbb{R}_+^3} \mathbf{1}[T(\theta, \theta')>x] \mathbf{1}[\frac{dQ_1}{dP_1}(\theta)>y] \mathbf{1}[\frac{dQ_2}{dP_2}(\theta')>z] dxdydz \right] P(d\theta, d\theta') \\
&= \int_{\mathbb{R}_+^3} \left[ \iint_{\Theta^2} \mathbf{1}[T(\theta, \theta')>x, \frac{dQ_1}{dP_1}(\theta)>y, \frac{dQ_2}{dP_2}(\theta')>z] P(d\theta, d\theta') \right] dxdydz \\
&= \int_{\mathbb{R}_+^3} P ( \{(\theta, \theta'): T(\theta, \theta')>x\} \cap \{(\theta, \theta'): \frac{dQ_1}{dP_1}(\theta)>y, \frac{dQ_2}{dP_2}(\theta')>z\} ) dxdydz.
\end{align*}
By the same reason,
\begin{align*}
&\iint_{\Theta^2} T(\theta, \theta') Q^*(d\theta, d\theta') \\
&= \int_{\mathbb{R}_+^3} P ( \{(\theta, \theta'): T(\theta, \theta')>x\} \cap \{(\theta, \theta'): \frac{dQ^*_{1}}{dP_1}(\theta)>y, \frac{dQ^*_{2}}{dP_2}(\theta')>z\} ) dxdydz.
\end{align*}

Therefore, to prove the desired inequality $\iint_{\Theta^2} T dQ^* \leq \iint_{\Theta^2} T dQ$, it suffices to prove the following: for all $x, y, z \in \mathbb{R}_+$,
\begin{align}
&P ( \{(\theta, \theta'): T(\theta, \theta')>x\} \cap \{(\theta, \theta'): \frac{dQ^*_{1}}{dP_1}(\theta)>y, \frac{dQ^*_{2}}{dP_2}(\theta')>z\} ) \notag \\
&\leq P ( \{(\theta, \theta'): T(\theta, \theta')>x\} \cap \{(\theta, \theta'): \frac{dQ_1}{dP_1}(\theta)>y, \frac{dQ_2}{dP_2}(\theta')>z\} ). \label{eq:rearrangement_ineq_proof_key}
\end{align}
To prove inequality \eqref{eq:rearrangement_ineq_proof_key}, fix $x, y, z \in \mathbb{R}_+$. Let
\[
\begin{array}{ll}
    U := \{ (\theta, \theta'): T(\theta, \theta') > x \} & \\
    A_1 := \{ \theta: \frac{dQ_1}{dP_1}(\theta)>y \} & \quad A_2 := \{ \theta': \frac{dQ_2}{dP_2}(\theta')>z \} \\
    A_1^* := \{ \theta: \frac{dQ^*_1}{dP_1}(\theta)>y \} & \quad A_2^* := \{ \theta': \frac{dQ_2^*}{dP_2}(\theta')>z \}.
\end{array}
\]
If $P(A_1 \times A_2) = 0$, inequality \eqref{eq:rearrangement_ineq_proof_key} holds because both sides are zero. Next, suppose $P(A_1 \times A_2) > 0$. We show that the hypothesis of Lemma \ref{lem:rearrangement_ineq_simple} holds. First, because $T$ is increasing, $U$ satisfies condition \eqref{eq:U_upper_set}. Also, recall from the proof of Proposition \ref{prop:rearrangement_ineq_IID} (i) that $Q^*_i$ is a decreasing rearrangement of $Q_i$. This implies that $P_i(A_i^*) = P_i(A_i)$, and $A_i^*$ is an interval with left endpoint zero. Thus, by Lemma \ref{lem:rearrangement_ineq_simple}, inequality \eqref{eq:rearrangement_ineq_proof_key} holds.

\end{proof}

\section{Proof of Theorem \ref{thm:methodology}} \label{appen:methodology}

Let $Q^* \in \mathcal{Q}^*$ be given. As argued in Section \ref{sec:methodology}, to prove Theorem \ref{thm:methodology}, it suffices to prove the following: for all $i$ and $\theta$ such that $(dQ^*_{i}/dP_{i})(\theta)>0$,
\begin{equation} \label{eq:methodology_proof_desired}
\int_\Theta t^X_i(\theta, \theta') Q^*(d\theta'|\theta) \geq \int_\Theta t^Y_i(\theta, \theta') Q^*(d\theta'|\theta).
\end{equation}
Note that condition $(dQ^*_{i}/dP_{i})(\theta)>0$ ensures that $Q^*(\cdot|\theta)$ is well-defined.

Fix $i$ and $\theta$. By WSCC, there exists $\hat \theta \in [0, 1]$ such that
\[
\theta' < \hat \theta \implies t_i^X(\theta, \theta') \geq t_i^Y(\theta, \theta') \quad \text{and} \quad \theta' > \hat \theta \implies t^X_i(\theta, \theta') \leq t^Y_i(\theta, \theta').
\]
Let $q^*: \Theta \rightarrow \mathbb{R}_+$ be the Radon-Nikodym derivative of $Q^*(\cdot|\theta)$ with respect to $P(\cdot|\theta)$. Since $dQ^*/dP$ decreases in each argument, $q^*(\theta')$ decreases in $\theta'$. Hence,
\begin{align*}
\int_\Theta [ t^X_i(\theta, \theta') - t^Y_i(\theta, \theta') ]^+ Q^*(d\theta'|\theta) &= \int_0^{\hat \theta} [ t^X_i(\theta, \theta') - t^Y_i(\theta, \theta') ]^+ q^*(\theta') P(d\theta'|\theta) \\
&\geq \int_0^{\hat \theta} [ t^X_i(\theta, \theta') - t^Y_i(\theta, \theta') ]^+ P(d\theta'|\theta) \cdot q^*(\hat \theta) \\
&= \int_\Theta [ t^X_i(\theta, \theta') - t^Y_i(\theta, \theta') ]^+ P(d\theta'|\theta) \cdot q^*(\hat \theta),
\end{align*}
where $z^+ := \max \{z, 0\}$. By the same reason,
\[
\int_\Theta [t^X_i(\theta, \theta')-t^Y_i(\theta, \theta')]^- Q^*(d\theta'|\theta) \leq \int_\Theta [t^X_i(\theta, \theta')-t^Y_i(\theta, \theta')]^- P(d\theta'|\theta) \cdot q^*(\hat \theta),
\]
where $z^- := \max \{-z, 0\}$. Thus,
\begin{align*}
&\int_\Theta t^X_i(\theta, \theta') Q^*(d\theta'|\theta) - \int_\Theta t^Y_i(\theta, \theta') Q^*(d\theta'|\theta) \\
&= \int_\Theta [ t^X_i(\theta, \theta') - t^Y_i(\theta, \theta') ]^+ Q^*(d\theta'|\theta) - \int_\Theta [t^X_i(\theta, \theta')-t^Y_i(\theta, \theta')]^- Q^*(d\theta'|\theta) \\
&\geq \left[ \int_\Theta t^X_i(\theta, \theta') P(d\theta'|\theta) - \int_\Theta t^Y_i(\theta, \theta') P(d\theta'|\theta) \right] \cdot q^*(\hat \theta) \geq 0,
\end{align*}
where the last inequality holds by RRC. This establishes inequality \eqref{eq:methodology_proof_desired}. $\square$

\section{Weak Single-Crossing Condition} \label{appen:scp}

In this section, we prove the equivalence between WSCC and condition \eqref{eq:scc_weak}. To show this, it is sufficient to show Proposition \ref{prop:WSCC_equivalence} below:

\begin{proposition} \label{prop:WSCC_equivalence}

Let $J, K: \Theta \rightarrow \mathbb{R}$. The following conditions are equivalent:

\noindent (i) There exists $\hat \theta \in [0, 1]$ such that for all $\theta'$,
\[
\theta' < \hat \theta \implies J(\theta') \geq K(\theta'), \quad \text{and} \quad \theta' > \hat \theta \implies J(\theta') \leq K(\theta').
\]

\noindent (ii) For all $\theta' > \theta''$,
\[
J(\theta'') < K(\theta'') \implies J(\theta') \leq K(\theta').
\]

\end{proposition}
\begin{proof} 

Without loss of generality, assume $K \equiv 0$.

\noindent \textbf{(i) $\Rightarrow$ (ii).} Suppose that $\theta'>\theta''$ and $J(\theta'')<0$. Condition (i) implies that $\theta'' \geq \hat \theta$. Since $\theta' > \theta'' \geq \hat \theta$, it follows by condition (i) that $J(\theta') \leq 0$.

\noindent \textbf{(ii) $\Rightarrow$ (i).} We divide into two cases.

\noindent \textit{Case 1: If $\{ \theta' \in \Theta: J(\theta') < 0 \} \neq \varnothing$.} In this case, define $\hat \theta := \inf \{ \theta' \in \Theta: J(\theta') < 0 \}$. Then, for $\theta' < \hat \theta$, we have $J(\theta') \geq 0$ by definition. Next, suppose $\theta' > \hat \theta$. By the property of the infimum, there exists $\theta'' \in [\hat \theta, \theta')$ such that $J(\theta'') < 0$. Condition (ii) implies that $J(\theta') \leq 0$.

\noindent \textit{Case 2: If $\{ \theta' \in \Theta: J(\theta') < 0 \} = \varnothing$.} In this case, $J(\theta') \geq 0$ for all $\theta'$. Hence, if we let $\hat \theta := 1$, then condition (i) holds.

\end{proof}

\section{Proof of Theorem \ref{thm:ranking}} \label{appen:ranking}

By Theorem \ref{thm:methodology}, to prove Theorem \ref{thm:ranking}, it suffices to show that the pairs $(X, Y) = (I, II), (I, A), (A, W), (II, W)$ satisfy WSCC and RRC. By the Revenue Equivalence Principle \citep{Myer81}, RRC holds. It remains to verify WSCC.

\noindent \textbf{(i) $(X, Y) = (I, II)$.} Given $i$ and $\theta$, let $\hat \theta = b^I(\theta)$. Then, since $b^I(\theta)<\theta$,
\begin{align*}
\text{for $\theta' < \hat \theta$}, \quad &t_i^I(\theta, \theta') = b^I(\theta) = \hat \theta > \theta' = t_i^{II}(\theta, \theta') \\
\text{for $\hat \theta < \theta' < \theta$}, \quad &t_i^I(\theta, \theta') = b^I(\theta) = \hat \theta < \theta' = t_i^{II}(\theta, \theta') \\
\text{for $\theta' = \theta$}, \quad &t_i^I(\theta, \theta') = b^I(\theta)/2 = \hat \theta/2 < \theta/2 = t_i^{II}(\theta, \theta') \\
\text{for $\theta' > \theta$}, \quad &t_i^I(\theta, \theta') = 0 = t_i^{II}(\theta, \theta'). \quad \square
\end{align*}

\noindent \textbf{(ii) $(X, Y) = (I, A)$.} Given $i$ and $\theta$, let $\hat \theta = \theta$. It is straightforward to show that $b^I(\theta)>b^A(\theta)$. Hence,
\begin{align*}
\text{for $\theta' < \theta$}, \quad &t_i^I(\theta, \theta') = b^I(\theta) > b^A(\theta) = t_i^A(\theta, \theta') \\
\text{for $\theta' > \theta$}, \quad &t_i^I(\theta, \theta') = 0 < b^A(\theta) = t_i^A(\theta, \theta'). \quad \square
\end{align*}

\noindent \textbf{(iii) $(X, Y) = (A, W)$.} Note first that
\begin{align} \label{eq:bW>bA}
b^W(\theta) &= \int_0^\theta z [-\log(1-F(z))]' dz = -\theta \log(1-F(\theta)) + \int_0^\theta \log(1-F(z)) dz \notag \\
&> \theta - \int_0^\theta F(z) dz = b^A(\theta),
\end{align}
where the third inequality holds because $-\log(1-z) > z$ for $z \in (0, 1)$. 

Now, let $i$ and $\theta$ be given. By inequality \eqref{eq:bW>bA} and continuity, there exists $0<\hat \theta<\theta$ such that $b^W(\hat \theta) = b^A(\theta)$. Then,
\begin{align*}
\text{for $\theta' < \hat \theta$}, \quad &t_i^A(\theta) = b^A(\theta) = b^W(\hat \theta) > b^W(\theta') = t_i^W(\theta, \theta') \\
\text{for $\hat \theta < \theta' < \theta$}, \quad &t_i^A(\theta) = b^A(\theta) = b^W(\hat \theta) < b^W(\theta') = t_i^W(\theta, \theta') \\
\text{for $\theta' \geq \theta$}, \quad &t_i^A(\theta) = b^A(\theta) < b^W(\theta) = t_i^W(\theta, \theta'). \quad \square
\end{align*}

\noindent \textbf{(iv) $(X, Y) = (II, W)$.} We proceed in three steps.

\noindent \textbf{Step 1.} \textit{$\lim_{\theta \rightarrow 0} (b^W)'(\theta) = \lim_{\theta \rightarrow 0} \theta f(\theta) / [1-F(\theta)] = 0$.} 

Suppose on the contrary that $\lim_{\theta \rightarrow 0} \theta f(\theta) / [1-F(\theta)] = L > 0$, where the limit exists by condition \eqref{eq:F_hazard}. Condition \eqref{eq:F_hazard} implies further that $\theta f(\theta) / [1-F(\theta)] \geq L$ for all $\theta$. Hence, for all $\theta_L, \theta_H$ with $0 < \theta_L < \theta_H < 1$, 
\begin{equation} \label{eq:II_W_step1}
\int_{\theta_L}^{\theta_H} \frac{f(\theta)}{1-F(\theta)} d\theta \geq \int_{\theta_L}^{\theta_H} \frac{L}{\theta} d\theta \implies -\log \frac{1-F(\theta_H)}{1-F(\theta_L)} \geq L \log \frac{\theta_H}{\theta_L}.
\end{equation}
However, if we take the limit $\theta_L \rightarrow 0$, the left-hand side converges to $- \log [1-F(\theta_H)]$, whereas the right-hand side diverges to infinity. Hence, inequality \eqref{eq:II_W_step1} cannot hold for sufficiently small values of $\theta_L$, a contradiction. $\blacksquare$

\medskip

\noindent \textbf{Step 2.} \textit{There exist $\theta^* \in (0, 1)$ such that}
\begin{equation} \label{eq:pf_II_W_theta_*}
\theta \leq \theta^* \implies b^W(\theta) \leq \theta, \quad \text{and} \quad \theta \geq \theta^* \implies b^W(\theta) \geq \theta.
\end{equation}

By definition, $b^W(0) = 0$. Also, by Step 1, $\lim_{\theta \rightarrow 0} (b^W)'(\theta) = 0 < 1$. It follows that for all $\theta$ sufficiently close to $0$, we have $b^W(\theta) < \theta$. Furthermore, \citet[Prop. 1]{Kris97} show that $\lim_{\theta \rightarrow 1} b^W(\theta) = \infty$. Hence, for all $\theta$ sufficiently close to $1$, we have $b^W(\theta) > \theta$. Hence, there exists an intersection $\theta^* \in (0, 1)$ satisfying $b^W(\theta^*)=\theta^*$. Because condition \eqref{eq:F_hazard} implies that $b^W$ is convex, property \eqref{eq:pf_II_W_theta_*} holds. $\blacksquare$

\medskip

\noindent \textbf{Step 3.} \textit{$(X, Y) = (A, W)$ satisfies WSCC.} Given $i$ and $\theta$, we divide into two cases.

\noindent \textit{Step 3-Case 1: If $\theta \leq \theta^*$.} Let $\hat \theta = \theta$. Then,
\begin{align*}
\text{for $\theta' < \hat \theta$,} \quad & t_i^{II}(\theta) = \theta' \geq b^W(\theta') = t_i^{W}(\theta, \theta') \\
\text{for $\theta' > \hat \theta$,} \quad & t_i^{II}(\theta, \theta') = 0 < b^W(\theta) = t_i^{W}(\theta, \theta').
\end{align*}

\noindent \textit{Step 3-Case 2: If $\theta > \theta^*$.} Let $\hat \theta = \theta^*$. Then,
\begin{align*}
\text{for $\theta'<\hat \theta$,} \quad &t_i^{II}(\theta, \theta') = \theta' \geq b^W(\theta') = t_i^W(\theta, \theta') \\
\text{for $\hat \theta < \theta' < \theta$,} \quad &t_i^{II}(\theta, \theta') = \theta' \leq b^W(\theta') = t_i^W(\theta, \theta') \\
\text{for $\theta' = \theta$,} \quad &t_i^{II}(\theta, \theta') = \theta/2 < \theta \leq b^W(\theta) = t_i^{W}(\theta, \theta') \\
\text{for $\theta' > \theta$,} \quad &t_i^{II}(\theta, \theta') = 0 < b^W(\theta) = t_i^W(\theta, \theta'). \quad \square
\end{align*}

\section{Proof of Proposition \ref{prop:linkage_4}} \label{appen:linkage_4}

We claim that $(X, Y) = (I, II), (I, A), (A, W), (II, W), (I, W)$ satisfy both conditions (i) and (ii), and the others satisfy neither. 

\noindent \textbf{Condition (i).} In the proof of Theorem \ref{thm:ranking}, we have already shown that $(X, Y) = (I, II), (I, A), (A, W), (II, W)$ satisfy WSCC. Also, a similar argument shows that $(X, Y) = (I, W)$ satisfies WSCC. It straightforward to check that the remaining pairs of auctions do not satisfy WSCC.

\noindent \textbf{Condition (ii).} First, \citet[Sec. 7.1-7.2]{Kris02} shows that $(X, Y)=(I, II)$ satisfies LC2 and $(X, Y)=(I, A)$ satisfies LC1.

Also, to see that $(X, Y) = (A, W)$ satisfies LC1, note that since $t_i^A(\theta, \theta')$ is constant in $\theta'$, we have $\partial_2 e_i^A(\theta, \theta) = 0$. However, since $t_i^W(\theta, \theta')$ increases in $\theta'$, affiliation implies that $\partial_2 e_i^W(\theta, \theta) \geq 0$. Hence, $(X, Y) = (A, W)$ satisfies LC1.

Next, we show that $(X, Y) = (II, W)$ satisfies LC1. When bidder $i$ has type $\theta$, denote the cumulative distribution and the probability density of the competitor's type as $F(\cdot|\theta)$ and $f(\cdot|\theta)$. Also, let $\lambda(\cdot|\theta)$ be the hazard rate:
\[
\lambda(z|\theta) := \frac{f(z|\theta)}{1-F(z|\theta)}.
\]
It is well-known that $\lambda(z|\theta)$ decreases in $\theta$ \citep[Fact 3]{Kris97}.

\citet[proof of Thm. 3]{Kris97} show that
\[
e_i^{II}(\hat \theta, \theta) = \int_0^{\hat \theta} z f(z|\theta) dz, \quad \text{and} \quad e_i^W(\hat \theta, \theta) = \int_0^{\hat \theta} z \lambda(z|z) [1-F(z|\theta)] dz.
\]
Hence, to prove that $(X, Y) = (II, W)$ satisfies LC1, it suffices to show that
\begin{equation} \label{eq:II_W_LC_desired}
\frac{\partial}{\partial \theta} f(z|\theta) \leq \lambda(z|z) \frac{\partial}{\partial \theta} [1-F(z|\theta)] \quad \text{for all $z \leq \theta$}.
\end{equation}
To prove inequality \eqref{eq:II_W_LC_desired}, note that because $\lambda(z|\theta)$ decreases in $\theta$,
\begin{equation} \label{eq:II_W_LC_step1}
    \lambda(z|\theta) \leq \lambda(z|z) \implies f(z|\theta) \leq \lambda(z|z) [1-F(z|\theta)].
\end{equation}
Also, by the same reason, for $\epsilon > 0$,
\[
 \lambda(z|\theta+\epsilon) \leq \lambda(z|\theta) \Rightarrow \frac{f(z|\theta+\epsilon)-f(z|\theta)}{\epsilon \cdot f(z|\theta)} \leq \frac{[1-F(z|\theta+\epsilon)]-[1-F(z|\theta)]}{\epsilon \cdot [1-F(z|\theta)]}.
\]
Taking the limit $\epsilon \rightarrow 0$ yields
\begin{equation} \label{eq:II_W_LC_step2}
    \frac{\frac{\partial}{\partial \theta} f(z|\theta)}{f(z|\theta)} \leq \frac{\frac{\partial}{\partial \theta} [1-F(z|\theta)]}{1-F(z|\theta)}.
\end{equation}
Inequalities \eqref{eq:II_W_LC_step1}-\eqref{eq:II_W_LC_step2} imply the desired inequality \eqref{eq:II_W_LC_desired}.

Finally, $(X, Y) = (I, W)$ satisfies LC1 because $(X, Y) = (I, A), (A, W)$ satisfy LC1. It is easy to verify that the remaining pairs satisfy neither LC1 nor LC2. $\square$

\bibliographystyle{elsarticle-harv}
\bibliography{bib_draft}

\end{document}